\documentclass[11pt]{article}

\pdfoutput=1

\usepackage[
  bookmarks=true,
  bookmarksnumbered=true,
  bookmarksopen=true,
  pdfborder={0 0 0},
  breaklinks=true,
  colorlinks=true,
  linkcolor=black,
  citecolor=black,
  filecolor=black,
  urlcolor=black,
]{hyperref}
\usepackage{amssymb,amsmath,amsthm}
\usepackage[round]{natbib}
\renewcommand{\cite}{\citep}
\usepackage[margin=1in]{geometry}

\usepackage{ifthen}
\usepackage{xspace}
\usepackage{amsmath}
\usepackage{algorithm}
\usepackage[noend]{algorithmic}

\newcommand{\LINECOMMENT}[1]{\STATE \#\ #1}

\usepackage[capitalize]{cleveref}
\usepackage{nicefrac}
\usepackage{mathtools}

\theoremstyle{plain}
\newtheorem{theorem}{Theorem}[section]
\newtheorem{lemma}[theorem]{Lemma}
\newtheorem{proposition}[theorem]{Proposition}
\newtheorem{remark}{Remark}
\newtheorem{open}[theorem]{Open Problem}

\theoremstyle{definition}
\newtheorem{definition}[theorem]{Definition}

\DeclareMathOperator{\argmax}{argmax}

\DeclareSymbolFont{AMSb}{U}{msb}{m}{n}
\DeclareMathSymbol{\N}{\mathord}{AMSb}{"4E}
\DeclareMathSymbol{\Z}{\mathord}{AMSb}{"5A}
\DeclareMathSymbol{\R}{\mathord}{AMSb}{"52}

\providecommand{\e}{\ensuremath{\mathrm{e}}\xspace}

\providecommand{\SetCard}[1]{\ensuremath{| #1 |}\xspace}
\providecommand{\SET}[1]{\ensuremath{\{ #1 \}}\xspace}
\providecommand{\Abs}[1]{\ensuremath{| #1 |}\xspace}

\providecommand{\PROB}{\ensuremath{{\rm Prob}}\xspace}
\providecommand{\Prob}[2][]{\ensuremath{%
\ifthenelse{\equal{#1}{}}{\PROB[#2]}{\PROB_{#1}[#2]}}\xspace}
\providecommand{\ProbC}[3][]{\Prob[#1]{#2\;|\;#3}}
\providecommand{\Expect}[2][]{\ensuremath{%
\ifthenelse{\equal{#1}{}}{\mathbb{E}}{\mathbb{E}_{#1}}%
\left[#2\right]}\xspace}

\providecommand{\Event}[2][]{\ensuremath{\ifthenelse{\equal{#1}{}}{\mathcal{#2}}{\mathcal{#2}_{{#1}}}}\xspace}

\providecommand{\SetCardB}[1]{\ensuremath{\bigl| #1 \bigr|}\xspace}
\providecommand{\SETB}[1]{\ensuremath{\bigl\{ #1 \bigr\}}\xspace}
\providecommand{\SetB}[2]{\ensuremath{\SETB{#1 ~\big|~ #2}}\xspace}
\providecommand{\ProbB}[2][]{\ensuremath{%
\ifthenelse{\equal{#1}{}}{\PROB\bigl[#2\bigr]}{\PROB_{#1}\bigl[#2\bigr]}}\xspace}
\providecommand{\ProbCB}[3][]{\ProbB[#1]{#2\;\big|\;#3}}
\providecommand{\AbsB}[1]{\ensuremath{\bigl| #1 \bigr|}\xspace}

\newcommand{\Indexed}[2]{\ensuremath{%
\ifthenelse{\equal{#1}{}}{#2}{%
\ifthenelse{\equal{#1}{'}}{#2'}{%
\ifthenelse{\equal{#1}{''}}{#2''}{%
\ifthenelse{\equal{#1}{'''}}{#2'''}{%
\ifthenelse{\equal{#1}{^}}{\hat{#2}}{%
\ifthenelse{\equal{#1}{-}}{\bar{#2}}{%
{#2}_{#1}}}}}}}}\xspace}

\newcommand{\WOMEN}{\ensuremath{W}\xspace}
\newcommand{\MEN}{\ensuremath{M}\xspace}
\newcommand{\WORKERS}{\ensuremath{W}\xspace}
\newcommand{\FIRMS}{\ensuremath{F}\xspace}
\newcommand{\AGENTS}{\ensuremath{A}\xspace}

\newcommand{\WOMAN}[1][]{\ensuremath{\Indexed{#1}{w}}\xspace}
\newcommand{\MAN}[1][]{\ensuremath{\Indexed{#1}{m}}\xspace}
\newcommand{\WORKER}[1][]{\ensuremath{\Indexed{#1}{w}}\xspace}
\newcommand{\FIRM}[1][]{\ensuremath{\Indexed{#1}{f}}\xspace}
\newcommand{\AG}[1][]{\ensuremath{\Indexed{#1}{a}}\xspace}
\newcommand{\UNMATCHED}{\ensuremath{\emptyset}\xspace}

\newcommand{\RWOMAN}[1][]{\ensuremath{\Indexed{#1}{X}}\xspace}
\newcommand{\RSET}[1]{\ensuremath{\mathcal{S}_{#1}}\xspace}
\newcommand{\CHAINS}[2][]{\ensuremath{%
\ifthenelse{\equal{#1}{}}{b_{#2}}{b_{#2}^{(#1)}}}\xspace}
\newcommand{\ALG}{\ensuremath{\mathcal{A}}\xspace}

\newcommand{\PREF}[2][]{\ensuremath{%
\ifthenelse{\equal{#1}{}}{\succ_{#2}}{\succ^{(#1)}_{#2}}}\xspace}
\newcommand{\Pref}[4][]{\ensuremath{#3 \PREF[#1]{#2} #4}\xspace}

\newcommand{\PREFH}{\ensuremath{\hat{\succ}}\xspace}
\newcommand{\PrefH}[2]{\ensuremath{#1 \PREFH #2}\xspace}
\newcommand{\PREFE}{\ensuremath{\succ}\xspace}
\newcommand{\PrefE}[2]{\ensuremath{#1 \PREFE #2}\xspace}
\newcommand{\PREFT}[1]{\ensuremath{\succ^*_{#1}}\xspace}
\newcommand{\PrefT}[3]{\ensuremath{#2 \PREFT{#1} #3}\xspace}

\newcommand{\RANKINGS}{\ensuremath{\mathcal{P}}\xspace}
\newcommand{\Rankings}[2][]{\ensuremath{%
\ifthenelse{\equal{#1}{}}{\mathcal{P}_{#2}}{\mathcal{P}^{(#1)}_{#2}}}\xspace}
\newcommand{\NumRankings}[2][]{\ensuremath{%
\ifthenelse{\equal{#1}{}}{r_{#2}}{\rho^{(#1)}_{#2}}}\xspace}
\newcommand{\NUMRANKINGS}[1][]{\ensuremath{%
\ifthenelse{\equal{#1}{}}{\rho}{\rho^{(#1)}}}\xspace}

\newcommand{\Capacity}[1]{\ensuremath{q_{#1}}\xspace}
\newcommand{\MaxCapacity}{\ensuremath{q}\xspace}

\newcommand{\GMATCH}[1][]{\ensuremath{%
\ifthenelse{\equal{#1}{}}{\mu}{\mu^{(#1)}}}\xspace}
\newcommand{\MATCH}[1][]{\ensuremath{%
\ifthenelse{\equal{#1}{}}{\mu}{\mu^{(#1)}}}\xspace}
\newcommand{\Partners}[2][]{\ensuremath{%
\ifthenelse{\equal{#1}{}}{\mu(#2)}{\mu^{(#1)}(#2)}}\xspace}
\newcommand{\Partner}[2][]{\ensuremath{%
\ifthenelse{\equal{#1}{}}{\mu(#2)}{\mu^{(#1)}(#2)}}\xspace}
\newcommand{\BLOCKING}[1][]{\ensuremath{%
\ifthenelse{\equal{#1}{}}{B}{B^{(#1)}}}\xspace}
\newcommand{\BWORKER}[1]{\ensuremath{w^{(#1)}}\xspace}
\newcommand{\BFIRM}[1]{\ensuremath{f^{(#1)}}\xspace}
\newcommand{\BWOMAN}[1]{\ensuremath{w^{(#1)}}\xspace}
\newcommand{\BMAN}[1]{\ensuremath{m^{(#1)}}\xspace}

\newcommand{\Constraint}[2]{\ensuremath{\langle #1, #2 \rangle}\xspace}
\newcommand{\CONSTRAINTS}{\ensuremath{\mathcal{C}}\xspace}
\newcommand{\Constraints}[2][]{\ensuremath{%
\ifthenelse{\equal{#1}{}}{\CONSTRAINTS_{#2}}{\CONSTRAINTS^{(#1)}_{#2}}}\xspace}
\newcommand{\SamplePO}{\ensuremath{\mathcal{S}}\xspace}

\newcommand{\GaleShapley}[1]{\ensuremath{\text{\textsc{Gale-Shapley}}\bigl(#1\bigr)}\xspace}
\newcommand{\SampleRanking}{\ensuremath{\text{\textsc{Sample-Representative-Preference}}}\xspace}
\newcommand{\SampleRankingMany}{\ensuremath{\text{\textsc{Representative-Preference-Large-Constraints}}}\xspace}
\newcommand{\SampleSampleRankingMany}{\ensuremath{\text{\textsc{Sample-Representative-Preference-Large-Constraints}}}\xspace}

\newcommand{\PrefFrac}[3][]{\ensuremath{%
\ifthenelse{\equal{#1}{}}{p_{#2,#3}}{p_{#2,#3}(#1)}}\xspace}
\newcommand{\PrefFracE}[2]{\ensuremath{\hat{p}_{#1,#2}}\xspace}

\newcommand{\RemElements}[1][]{\ensuremath{%
\ifthenelse{\equal{#1}{}}{U}{U^{(#1)}}}\xspace}
\newcommand{\LastElem}[1][]{\ensuremath{%
\ifthenelse{\equal{#1}{}}{\bar{z}}{\bar{z}^{(#1)}}}\xspace}
\newcommand{\LastFrac}[3][]{\ensuremath{%
\ifthenelse{\equal{#1}{}}{\ell_{#2,#3}}{\ell_{#2,#3}(#1)}}\xspace}
\newcommand{\LastFracE}[3][]{\ensuremath{%
\ifthenelse{\equal{#1}{}}{\hat{\ell}_{#2,#3}}{\hat{\ell}_{#2,#3}(#1)}}\xspace}

\newcommand{\Closure}[2]{\ensuremath{\overline{#2}^{#1}}\xspace} 

\hypersetup{
  pdfauthor      = {Ehsan Emamjomeh-Zadeh <ehsan7069@gmail.com>, Yannai A. Gonczarowski <yannai@gonch.name>, and David Kempe <david.m.kempe@gmail.com>},
  pdftitle       = {The Complexity of Interactively Learning a Stable Matching by Trial and Error},
}
\begin{document}

\title{The Complexity of\texorpdfstring{\\}{ }Interactively Learning a Stable Matching by Trial and Error\thanks{A one-page abstract of this paper appeared in the Proceedings of the 21st ACM Conference on Economics and Computation (EC 2020). A one-minute flash video for this paper is available at: \mbox{\url{https://www.youtube.com/watch?v=9GjbPa24DwE}} . The authors thank Yu Cheng, Scott Duke Kominers, Noam Nisan, Assaf Romm, and Ran Shorrer for useful discussions, and anonymous reviewers for EC 2020 for helpful feedback. EE and DK were supported in part by NSF Grant IIS-1619458 and ARO MURI grant W911NF1810208.}~\thanks{In proof: after this paper appeared at EC'20, the authors of \citep{bei:chen:zhang:complexity} contacted us (personal communication, August 20, 2020) and made us aware of their paper, of which we had been unaware until that point. As they pointed out, some of the results that we have proven in the current paper were either already proven or similar to results already proven in their paper. See the newly added \cref{bei-chen-zhang} following our introduction for more details.}}
\date{September 16, 2020}

\author{Ehsan Emamjomeh-Zadeh\thanks{University of Southern California; \emph{e-mail}: \href{mailto:ehsan7069@gmail.com}{ehsan7069@gmail.com}.} \and Yannai A.\ Gonczarowski\thanks{Microsoft Research; \emph{e-mail}: \href{mailto:yannai@gonch.name}{yannai@gonch.name}.} \and David Kempe\thanks{University of Southern California; \emph{e-mail}: \href{mailto:david.m.kempe@gmail.com}{david.m.kempe@gmail.com}.}}

\maketitle

\begin{abstract}
In a stable matching setting, we consider a query model that allows for an interactive learning algorithm to make precisely one type of query: proposing a matching, the response to which is either that the proposed matching is stable, or a blocking pair (chosen adversarially) indicating that this matching is unstable. For one-to-one matching markets, our main result is an essentially tight upper bound of $O(n^2\log n)$ on the deterministic query complexity of interactively learning a stable matching in this coarse query model, along with an efficient randomized algorithm that achieves this query complexity with high probability. For many-to-many matching markets in which participants have responsive preferences, we first give an interactive learning algorithm whose query complexity and running time are polynomial in the size of the market if the maximum quota of each agent is bounded; our main result for many-to-many markets is that the deterministic query complexity can be made polynomial (more specifically, $O(n^3 \log n)$) in the size of the market even for arbitrary (e.g., linear in the market size) quotas.

\end{abstract}

\section{Introduction} \label{sec:introduction}
\paragraph{Problem and Background}

In the classic Stable Matching Problem \citep{gale:shapley}, there are $n$ \emph{women} and $n$ \emph{men}; each woman has a subset of men she deems acceptable and a strict preference order over them; similarly, each man has a subset of women he finds acceptable and a strict preference order over them. The goal is to find a \emph{stable matching}: a one-to-one mapping between mutually acceptable women and men that contains no \emph{blocking pair}, i.e., no pair of a woman and a man who mutually prefer each other over their current partner in the matching (or over being unmatched, if one or both are unmatched). In their seminal paper, \citet{gale:shapley} constructively proved that such a stable matching always exists.

The study of stable matchings within computer science was originated in \citeauthor{knuth:stable}'s famous series-of-lectures-turned-book \citep{knuth:stable}, in which Knuth focused on the stable matching problem as --- according to his foreword to the book --- an ideal introductory problem for the analysis of algorithms.
In this book, \citeauthor{knuth:stable}'s first order of business was that the most natural ``first attempt'' at an algorithm for finding a stable matching fails: if one repeatedly tries to resolve blocking pairs by matching the two participants of the pair with each other and matching their previous partners with each other, this may lead to a cycle.

Some decade and a half following \citeauthor{knuth:stable}'s book, \citet{roth:vandevate:paths} showed a positive result for a close variant of the algorithm, in which a blocking pair is resolved by matching the blocking pair with each other and leaving their previous partners unmatched; we refer to this algorithm as the \emph{myopic} algorithm.\footnote{Recall that Knuth's result is phrased for a slightly different version of this algorithm, which matches the two previous partners of the members of the blocking pair. Knuth's result, nonetheless, immediately implies the same result for the myopic algorithm of \citeauthor{roth:vandevate:paths}. It should be pointed out that without this change to the algorithm, the result of \citeauthor{roth:vandevate:paths} does not hold: there exists a matching from which there is no such ``path'' to a stable matching \citep{tamura:stable}.} \citeauthor{roth:vandevate:paths} showed that if the blocking pair to be resolved is chosen \emph{uniformly at random} among all blocking pairs (rather than adversarially),
then with probability one, the algorithm terminates in finite time, and thus yields a stable matching. \citeauthor{roth:vandevate:paths}'s result was later generalized by \citet{kojima:unver:paths} to the canonical ``many to many'' extension of the stable matching problem, in which each participant on each side of the market (the two sides of the market in this setting are customarily called \emph{firms} and \emph{workers}) can be matched to more than one participant on the other side, (e.g., up to a given \emph{quota}).\footnote{Generalizations to other models have been studied as well; see, e.g., \citet{chung:roommate,diamantoudi:miyagawa:xue:paths,klaus:klijn:paths}. See also \citet{ma:stable} for a related negative result.}

Another decade and a half later,
\citet{ackermann:goldberg:uncoordinated} focused on the number of iterations it takes the myopic algorithm to reach a stable matching, proving the following negative result:
for some preference structure, while convergence with probability one is assured,
it may take exponentially many iterations to reach a stable matching.
Their result applies already in the one-to-one setting, and thus of course extends to the many-to-many case as well.
  
Through the lens of learning theory, we can view this type of interaction as corresponding naturally to \citeauthor{angluin:queries-concept}'s \emph{Equivalence Query model} \citep{angluin:queries-concept} for learning a binary classifier, and its recent generalization due to \citet{OnlineLearning}.
In this interaction model, the learner proposes a candidate structure (here: a matching), and --- unless the proposal was correct (here: stable) --- learns of one local mistake (here: a blocking pair).
Because the queries are very coarse-grained compared to standard preference queries --- after all, a query comprises an entire matching and cannot target a specific comparison or even a specific participant --- we refer to the query model as the \emph{coarse query model}.
Viewed through this lens, the preceding results state that in the coarse query model, the myopic algorithm could take exponentially long to converge given randomly drawn (legal) responses, and might never converge given adversarial (legal) responses.
Taking this learning-theoretic perspective as our point of departure, in this paper, we ask whether more sophisticated \emph{interactive learning} algorithms in the coarse query model can be made to perform much better than the myopic algorithm, and if so, precisely to what extent.

Beyond a learning-theoretic interest,
our main question also has implications from a ``purely EconCS'' perspective. A famous quote by Kamal Jain \citep[see, e.g.,][]{papadimitriou:nash} questions the validity of market equilibria as a practical solution concept on the grounds that ``if your laptop cannot find it, neither can the market.'' This quote refers to a line of work that eventually culminated in showing that even finding approximate equilibria is PPAD-complete (and thus likely to require an exorbitant number of queries), not only in a decentralized manner, but even by a capable central planner. Most famously, this was shown for Nash equilibria \citep{daskalakis:nash,chen:deng:settling,rubinstein:approximate:nash}, and for Arrow-Debreu market equilibria \citep{chen:deng:settling,chen:dai:du:teng:settling,chen:paparas:yannakakis:markets,rubinstein:nash-inapproximability}, even with severely restricted utility functions.
A negative answer to our main question would therefore cast doubt on the possibility that a market designer with access only to coarse queries such as those of ``trial and error''\footnote{Let alone a decentralized market in which participants have limited information about their own preferences due to the costliness of learning them \cite[see, e.g.,][]{grenet:he:kuebler:decentralizing,narita:mismatch,hassidim:marciano:rommm:shorrer:truthful}.} could find a stable matching without the ability to ask each participant detailed questions about their preferences (like those that the Gale-Shapley algorithm requires). On the other hand, thinking of a blocking pair as ``minimal informative evidence'' of the instability of a matching, which is arguably what is naturally observed as a first sign of a matching being about to unravel\footnote{Whether due to the matching not being stable to begin with, due to preferences of participants changing, or due to participants learning more about their preferences following the announcement of the matching \cite{narita:mismatch}.}, a positive answer to our question would mitigate to some extent the negative message of \citet{ackermann:goldberg:uncoordinated}, providing evidence that stability can in fact be achieved even by a market designer looking at the market through a far coarser and less informative lens than usually considered.

\paragraph{One-to-one markets}
We start our analysis in the one-to-one setting of men and women, and consider adversarial responses to the queries of our algorithm. (The model is defined precisely in Section~\ref{sec:preliminaries}.) A trivial observation is that an exponential-query brute-force search over all matchings can find a stable matching. Some further consideration gives rise to the following algorithm: at every step, fix \emph{speculative} preferences for all participants; the only requirement is that these speculative preferences must be consistent with all of the information gathered so far about the participants' preferences. (In the first step, any preferences are consistent, as no information has been gathered yet.) The algorithm then proposes a matching that is stable with respect to these speculative preferences. If this matching is unstable, then the response to the query tells the algorithm that some pair $(w,m)$ is blocking; that is, that $w$ prefers $m$ to her assigned partner and that $m$ prefers~$w$ to his assigned partner. At least one of these two pairwise comparisons must contradict the speculative preferences and so constitutes new information upon which to base the new speculative preferences in the next step. Noticing that throughout the run of the algorithm, the algorithm can collect at most $\binom{n}{2}=O(n^2)$ pairwise comparisons for each participant's preferences, we conclude that this algorithm makes at most $O(n^3)$ unsuccessful queries, and therefore finds a stable matching in $O(n^3)$ queries (and in polynomial running time).

The above cubic algorithm essentially attempts to ``sort'' the preferences of each participant (unless it is ``interrupted'' by finding a stable matching). In the worst case, it uses quadratically many of its queries to ``sort'' each participant's preferences. While finding a stable matching in polynomially many queries is already appealing, encapsulating a quadratic-query sort within this algorithm is dissatisfying. The challenge with replacing this ``inner sort'' with an $O(n\log n)$-query sort is that unlike standard sorting algorithms, in the coarse query model, the next comparison is chosen not by the algorithm, but adversarially. It is therefore a priori unclear whether the query complexity of each such ``inner sort'' can be improved. Our first main result, proved in Section~\ref{sec:one-to-one}, is a positive answer to this question: it is possible, at the beginning of each step of the algorithm, to choose the speculative preferences in a way such that no more than $O(n\log n)$ comparisons are ever made for any one participant, implying the following theorem:

\begin{theorem}[Informal, see \cref{thm:learning-matching-iterations}\footnote{In proof: after this paper appeared at EC'20, the authors of \citep{bei:chen:zhang:complexity} contacted us (personal communication, August 20, 2020) and made us aware of their paper, of which we had been unaware until that point. As they pointed out, their paper contains a proof of \cref{thm:learning-matching-iterations}. While of the same order, the number of queries that our algorithm uses is smaller by a constant factor, and our proof is somewhat shorter, since we used \citeauthor{yu:proportional-transitivity}'s Proportional Transitivity Theorem, whereas they proved the result from first principles; however, neither this difference nor the difference in proofs is substantial. See the newly added \cref{bei-chen-zhang} following our introduction for more details.}]
\label{thm:informal-query-complexity-one-to-one}
In the coarse query model, the deterministic \emph{query} complexity of finding a stable matching is $O(n^2\log n)$.
\end{theorem}

The main combinatorial tool that we leverage in our proof of this theorem is \citeauthor{yu:proportional-transitivity}'s remarkable Proportional Transitivity Theorem \citep{yu:proportional-transitivity}.
Implementing the resulting strategy for finding the ``right'' speculative preferences to fix is computationally non-trivial;
indeed, note that \cref{thm:informal-query-complexity-one-to-one} only expresses guarantees about the number of rounds of interaction, but not about the computational complexity.
Fortunately, a Markov chain-based algorithm due to \citet{huber:linear-extensions} for sampling linear extensions of a partial order can be leveraged to obtain ``good'' speculative preferences efficiently (i.e., in running time polynomial in $n$), up to negligible probability of failure:

\begin{theorem}[Informal, see \cref{thm:learning-matching-time}\footnote{In proof: after this paper appeared at EC'20, the authors of \citep{bei:chen:zhang:complexity} contacted us (personal communication, August 20, 2020) and made us aware of their paper, of which we had been unaware until that point. As they pointed out, their paper contains a proof of \cref{thm:learning-matching-time}. While of the same order, the number of queries that our algorithm uses is smaller by a constant factor, and our proof is somewhat shorter, since we used \citeauthor{yu:proportional-transitivity}'s Proportional Transitivity Theorem, whereas they proved the result from first principles; however, neither this difference nor the difference in proofs is substantial. See the newly added \cref{bei-chen-zhang} following our introduction for more details.}]
In the coarse query model, a stable matching can be \emph{efficiently} found with high probability using $O(n^2\log n)$ queries.
\end{theorem}

Can this quadratic (up to the $\log$ factor) query complexity be improved? The existing literature on the query and communication complexity of stable matchings \citep{ng:hirschberg:stable,chou:lu:stable,segal:communication,gonczarowski:nisan:ostrovsky:rosenbaum} provides some very robust lower bounds on the complexity of finding a stable matching. However, these are all obtained in a model in which each query is to the preference list of one participant, or more generally involves only preference lists of participants on one of the sides of the market.\footnote{Those papers at their core are interested in multi-party computation of a stable matching from distributed input, rather than learning a stable matching by an outside agent who can query the entire market (albeit in a very limited way).} The most general of these lower bounds \citep{gonczarowski:nisan:ostrovsky:rosenbaum} states that $\Omega(n^2)$ queries are required to find a stable matching. While queries in that model are very flexible, the theorem of \citet{gonczarowski:nisan:ostrovsky:rosenbaum} that proves this bound also proves that in that model, even verifying the stability of a proposed matching requires quadratically many queries. This is in contrast with our coarse query model, in which verifying the stability of a proposed matching takes one query. In general, though, queries in any of the models of those papers are more powerful and flexible than in our coarse query model, in the sense that in those models, a query can specify, for example, a precise pairwise comparison to be revealed. By contrast, queries in our coarse query model are answered adversarially.\footnote{Thus, the relation between the models is that one query in our coarse query model can be simulated by $n^2$ queries in the models of those papers, but no number of queries in our coarse query model suffices to simulate a single query in any of their models. Therefore, an upper bound in any of the models of those papers cannot be translated to any upper bound in our coarse query model, and an $\Omega(n^2)$ lower bound in any of the models of those papers translates only to a trivial $\Omega(1)$ lower bound in our coarse query model.} Given these differences between the models, it is quite curious that in both models, the best known upper bound (up to $\log$ factors) is quadratic. Completing the picture, we present a quadratic lower bound in our coarse query model as well:

\begin{proposition}[Informal, see \cref{prop:lower-bound}\footnote{In proof: after this paper appeared at EC'20, the authors of \citep{bei:chen:zhang:complexity} contacted us (personal communication, August 20, 2020) and made us aware of their paper, of which we had been unaware until that point. As they pointed out, their paper contains a theorem similar to \cref{prop:lower-bound}. Our result is somewhat stronger by allowing preferences on one side to be known and identical, and preferences on the other side to be i.i.d.\ random, whereas theirs is a worst-case result; however, the qualitative message is similar. See the newly added \cref{bei-chen-zhang} following our introduction for more details.}]
  In the coarse query model, the query complexity of finding a stable matching is $\Omega(n^2)$. This holds even for the number of expected queries by a randomized algorithm for surely finding a stable matching when the preferences of participants on one side of the market are identical and known, and the preferences of the participants on the other side are i.i.d.\ uniform.
\end{proposition}

The proof of this lower bound is less elaborate than those of the quadratic lower bounds of the above-mentioned papers. This provides further evidence regarding the weakness/coarseness of the queries that are allowed in our model, in some sense making the $O(n^2 \log n)$ upper bound more surprising. The remaining gap between our upper and lower bounds, while in a very different model than that in which \citet{gonczarowski:nisan:ostrovsky:rosenbaum} presented their gap, is between the exact same bounds of $\Omega(n^2)$ and $O(n^2 \log n)$. Both gaps, in very different technical senses, may be thought of as leaving open the interesting question of whether the number of ``atomic'' queries required to find a stable matching can be significantly less than the $\Theta(n^2\log n)$ bits of information that are needed to encode the preference lists of all participants.

\paragraph{Many-to-many markets}
Finally, in Section~\ref{sec:many-to-many}, we generalize our results not only to many-to-one markets, but in fact ``all the way'' to many-to-many matching markets \citep{roth:job:matching}. We consider a market with firms on one side and workers on the other, where each participant has \emph{responsive preferences} \citep{roth:responsive}. That is, each participant has a strict preference order over the participants on the other side of the market, as well as a \emph{quota}: the maximum number of participants from the other side of the market to which this participant can be matched. The coarse query model that we consider in this setting is similar to the coarse query model from the one-to-one setting: an interactive learning algorithm may propose a matching, and the response is either that this matching is (pairwise) stable, or a blocking pair. However, the revelation of a blocking pair only reveals that the firm prefers the worker to \emph{some} worker it is currently matched with (but does not identify this worker), and that the worker prefers the firm to \emph{some} firm she is currently matched with (but does not identify this firm). Thus, the information given to the algorithm is even coarser than in the one-to-one matching case. It would of course also be natural --- and the appropriate choice for some applications --- to consider a model that reveals the discarded firm/worker; as we discus in \cref{many-many-with-discard-information}, our analysis from the one-to-one case applies to such a model without any need for adjustment. We focus on the even coarser model that does not reveal the discarded partners in order to keep in line with our motivation above: focusing on the weakest/coarsest model in which the algorithm still naturally obtains information about the underlying preferences.\footnote{Had the response specified the identity of the firm to which the blocking worker prefers the blocking firm (and vice versa), an approach identical to the one from the one-to-one case would have sufficed. Keeping this firm/worker unnamed breaks the well-known isomorphism \citep{roth:sotomayor:book} between one-to-one and many-to-one markets: even many-to-one markets already exhibit the full difficulty of the learning problem in many-to-many markets.}
In this model, a single response does not definitively imply even a single pairwise comparison in the preferences of any participant. It therefore becomes much less clear whether or how exponentially many queries (and exponential running time) can be avoided. Indeed, this problem bears a considerable resemblance to a problem that we show to be NP-hard (see \cref{prop:NP-hardness-topological}). Some further thought, though, shows that the above $O(n^3)$ algorithm for the one-to-one setting can be conceptually generalized to obtain an $n^{O(\MaxCapacity)}$ algorithm, where $\MaxCapacity \leq n$ is the maximum quota of any participant:

\begin{theorem}[Informal, see \cref{thm:iterations-simple-choice}]
In the coarse query model, the deterministic query complexity of finding a stable matching is $n^{O(\MaxCapacity)}$. Furthermore, a stable matching can be found in time $n^{O(\MaxCapacity)}$ using such a number of queries.
\end{theorem}

For a market with small (constant) quotas, this theorem is satisfactory. For markets in which at least one quota is large (e.g., superpolylogarithmic in the size of the market; for instance, where at least one employer employs a sizable chunk of the market), the query complexity of this theorem is too high, a result of the significantly coarser information the algorithm receives. Our main result for the many-to-many setting is nonetheless a positive one: it is possible to once again carefully set the speculative preferences in each step of the algorithm to obtain a polynomial number of queries:

\begin{theorem}[Informal, see \cref{thm:learning-matching-many-iterations}]
In the coarse query model, the deterministic query complexity of finding a stable matching is $O(n^3\log n)$ (regardless of the maximum quota $\MaxCapacity$).
\end{theorem}

This theorem makes no guarantees regarding the running time of an algorithm that makes this number of queries. Indeed, the main problem that we leave open is designing an interactive algorithm for many-to-many markets that runs in time polynomial in $n$, regardless of $\MaxCapacity$. As discussed in Section~\ref{sec:conclusions}, we identify intimate connections between this open problem and open problems regarding rapid mixing of certain stochastic processes under constraints that cease to be convex when $\MaxCapacity$ is greater than one.

\begin{open}
For the coarse query model, construct an interactive learning algorithm that finds a stable matching in many-to-many, or even many-to-one, markets with an expected polynomial number of queries and an expected polynomial run time.
\end{open}

As mentioned above, the coarse learning model we study can be viewed as a case of the generalized equivalence query model of \citet{OnlineLearning}. Note that the structures to be learned in our setting are the agents' \emph{preferences}, not the matching itself; the proposed matchings serve as conduits for receiving additional information about the preferences.
\citet{OnlineLearning} provided a general approach, based on a technique of \citet{BinarySearch}, for learning such structures in a number of rounds that is logarithmic in the number of candidate structures. Had this technique been applicable, their technique would have implied our $O(n^2 \log n)$ bound from \cref{thm:learning-matching-iterations}. However, their framework is \emph{provably} not applicable to our setting, necessitating the self-contained algorithm and analysis we provide here. We prove this fact in Appendix~\ref{sec:binary-search}.

\subsection{In Proof: Comparison with \citet{bei:chen:zhang:complexity}}\label{bei-chen-zhang}
After this paper appeared at EC'20, the authors of \citep{bei:chen:zhang:complexity} contacted us (personal communication, August 20, 2020) and made us aware of their paper, of which we had been unaware until that point. As they pointed out, in their paper, they asked a similar question to ours, for a number of domains including one-to-one matching markets, and proved theorems equivalent to our \cref{thm:learning-matching-iterations,thm:learning-matching-time} and a theorem similar to our \cref{prop:lower-bound}.
For \cref{thm:learning-matching-iterations,thm:learning-matching-time}, while of the same order, the number of queries that our algorithms use is smaller by a constant factor, and our proofs are somewhat shorter, since we used \citeauthor{yu:proportional-transitivity}'s Proportional Transitivity Theorem, whereas they proved the result from first principles; however, neither this difference nor the difference in proofs is substantial. For \cref{prop:lower-bound}, our result is somewhat stronger, by allowing preferences on one side to be known and identical, and preferences on the other side to be i.i.d.\ random, whereas theirs is a worst-case result; however, the qualitative message is similar. \citet{bei:chen:zhang:complexity} did not study many-to-one or many-to-many matching markets in their paper, and so our results for these markets, including \cref{thm:iterations-simple-choice,thm:learning-matching-many-iterations}, have no counterparts in their paper.

\section{Preliminaries} \label{sec:preliminaries}
\subsection{The Stable Matching Problem}

In a general many-to-many matching market, there is a set \WORKERS of $n_{\WORKERS}$ workers and a set \FIRMS of $n_{\FIRMS}$ firms.
We will use \WORKER and its variants exclusively for workers and \FIRM and its variants exclusively for firms. We write $\AGENTS = \WORKERS \cup \FIRMS$ for the set of all agents, and \AG and its variants when we define or reason about notation that applies equally to both sides of agents.
Each worker $\WORKER \in \WORKERS$ has \emph{quota} $\Capacity{\WORKER} \in \SET{1, \ldots, n_{\FIRMS}}$ and each firm $\FIRM \in \FIRMS$ has quota
$\Capacity{\FIRM} \in \SET{1, \ldots, n_{\WORKERS}}$.
Each worker $\WORKER \in \WORKERS$ has a \emph{preference order} (linear order) \PREFT{\WORKER} over a subset of \FIRMS, which we interpret as the subset of firms to which \WORKER prefers being matched over being unmatched. We write \PrefT{\WORKER}{\FIRM}{\FIRM[']} to denote that \WORKER (strictly) prefers \FIRM over \FIRM['], and write \PrefT{\WORKER}{\FIRM}{\UNMATCHED} to denote that \WORKER prefers \FIRM over being unmatched. Similarly, each firm $\FIRM \in \FIRMS$ has a preference order \PREFT{\FIRM} over a subset of \WORKERS.

A many-to-many matching $\GMATCH\subseteq\FIRMS\times\WORKERS$ is a set of worker-firm pairs such that the set of firms to which worker \WORKER is matched, $\Partners{\WORKER}=\SetB{\FIRM\in\FIRMS}{(\WORKER,\FIRM)\in\GMATCH}$, is of size at most $\Capacity{\WORKER}$, and the set of workers to whom firm \FIRM is matched, $\Partners{\FIRM}=\SetB{\WORKER\in\WORKERS}{(\WORKER,\FIRM)\in\GMATCH}$, is of size at most $\Capacity{\FIRM}$.

In the special case in which $\Capacity{\WORKER} = 1$ for all \WORKER and $\Capacity{\FIRM} = 1$ for all \FIRM, this is exactly the definition of a standard (not necessarily perfect) one-to-one matching. In this case, keeping with standard nomenclature, we refer to the two sides as \emph{women} \WOMEN and \emph{men} \MEN, using the notation \WOMAN and \MAN and their variants for individuals. In this case, we also write \Partner{\AG} for the (unique) partner of $\AG \in \AGENTS$. If \AG is unmatched, then we define $\Partner{\AG}=\UNMATCHED$.

Given a matching \GMATCH, an agent \AG is \emph{individually blocking} if she is matched to an agent that she does not prefer over being unmatched. A matching is \emph{individually rational} if no agent is individually blocking.
Given a matching \GMATCH, a pair $(\WORKER,\FIRM)\notin\GMATCH$ is called a \emph{blocking pair} (w.r.t.~\GMATCH)
if all of the following hold:
\begin{itemize}
\item
\PrefT{\WORKER}{\FIRM}{\UNMATCHED} and \PrefT{\FIRM}{\WORKER}{\UNMATCHED}.
\item
Either $\SetCardB{\Partners{\FIRM}}<\Capacity{\FIRM}$ or there exists $\WORKER['] \in \Partners{\FIRM}$ such that \PrefT{\FIRM}{\WORKER}{\WORKER[']}.
\item
Either $\SetCardB{\Partners{\WORKER}}<\Capacity{\WORKER}$ or there exists $\FIRM['] \in \Partners{\WORKER}$ such that \PrefT{\WORKER}{\FIRM}{\FIRM[']}.
\end{itemize}
A blocking pair would prefer to add each other to their respective sets of matches (while removing \FIRM['] and \WORKER['] from these respective sets if these sets are full).
A matching $\GMATCH$ is (pairwise) \emph{stable} (with respect to the preferences \PREFT{\WORKER} and \PREFT{\FIRM} and quotas \Capacity{\WORKER} and \Capacity{\FIRM}) if it is individually rational and no pair is blocking it.
That is, if \GMATCH is stable, then it is individually rational, and furthermore, for each pair $(\WORKER,\FIRM) \notin\GMATCH$ with \PrefT{\WORKER}{\FIRM}{\UNMATCHED} and \PrefT{\FIRM}{\WORKER}{\UNMATCHED}, we have either $\SetCardB{\Partners{\WORKER}}=\Capacity{\WORKER}$ and \PrefT{\WORKER}{\FIRM[']}{\FIRM} for all $\FIRM['] \in \Partners{\WORKER}$, or $\SetCardB{\Partners{\FIRM}}=\Capacity{\FIRM}$ and \PrefT{\FIRM}{\WORKER[']}{\WORKER} for all $\WORKER['] \in \Partners{\FIRM}$;
in words, a matching is stable iff each agent is only matched to agents that it prefers over being unmatched, and furthermore for each unmatched pair $(\WORKER,\FIRM)$, either \WORKER is matched to her quota of partners and prefers all her current partners to \FIRM, or \FIRM it matched to its quota of partners and prefers all its current partners to \WORKER.

\subsection{The Interactive Learning Model}
Our goal is to design an algorithm to produce/learn a stable matching. The challenge is that \PREFT{\WORKER} and \PREFT{\FIRM} are unknown to the algorithm; nonetheless, the algorithm seeks to find a stable matching with respect to them, by gradually learning about the individuals' preferences.
The learning proceeds in rounds.
In each round $t$, the algorithm proposes a matching $\GMATCH[t]$.
If \GMATCH[t] is stable with respect to the (initially unknown) \PREFT{\WORKER} and \PREFT{\FIRM}, the process terminates.
Otherwise, the algorithm learns about \emph{one} blocking pair $\BLOCKING[t] = (\BWORKER{t},\BFIRM{t})$ or one individually blocking agent.
If there are multiple possible blocking pairs (and/or individually blocking agents), the one that is revealed to the algorithm is chosen \emph{adversarially}.
Based on this new information, the algorithm can then compute a new proposed matching \GMATCH[t+1], and so on.

\subsection{Reduction to Full Preference Lists}

We will say that a matching market has \emph{full preference lists} if a) All preference lists are full: each agent prefers every agent on the other side of the market to being unmatched, and b) Quotas are balanced: $\sum_{\WORKER \in \WORKERS} \Capacity{\WORKER} = \sum_{\FIRM \in \FIRMS} \Capacity{\FIRM}$. (In a one-to-one market this translates to $n_{\WOMEN}=n_{\MEN}$.) We say that a matching in such a market is \emph{perfect} if every participant is matched to her full quota of partners. In such markets, a matching is individually rational if and only if it is perfect. Restricting attention to perfect matchings, a matching is stable if and only if it is not blocked by any pair. When stating our learning algorithms and analyzing them, it will be convenient and reduce clutter to restrict attention to markets with full preference lists. Since we will design our algorithms to only output perfect matchings, we will only have to deal with blocking pairs (rather than individually blocking agents) as responses. Luckily, by a well-known reduction in stable matching theory, this restriction is without loss of generality. See \cref{sec:partial-to-full} for the details of the embedding of the general problem in the special case of full preference lists, as well as for a discussion and a word of caution on applying other well-known reductions within the context of our coarse query model.

\subsection{Additional Notation}

When $\Capacity{\WORKER} = \Capacity{\FIRM} = 1$ for the blocking pair $(\WORKER,\FIRM) = (\BWORKER{t},\BFIRM{t})$, the algorithm can infer that \PrefT{\WORKER}{\FIRM}{\Partner{\WORKER}} and \PrefT{\FIRM}{\WORKER}{\Partner{\FIRM}}.
In the more general case, the algorithm only learns that there exist $\FIRM['] \in \Partners{\WORKER}$ and $\WORKER['] \in \Partners{\FIRM}$ such that \PrefT{\WORKER}{\FIRM}{\FIRM[']} and \PrefT{\FIRM}{\WORKER}{\WORKER[']}.
In Section~\ref{sec:one-to-one}, we will briefly discuss a model in which the algorithm also learns the exact identity of \WORKER['] and \FIRM[']; unsurprisingly, in this case, the problem directly reduces to the case in which $\Capacity{\WOMAN} = \Capacity{\MAN} = 1$ for all \WOMAN and \MAN.\footnote{For many-to-many settings, this will be clear after we present our learning algorithm. For many-to-one settings, this can also be seen to follow from a well-known reduction from many-to-one to one-to-one matching markets that ``splits'' each firm with quota $>1$ into unit-demand firms. As hinted in \cref{sec:introduction}, this reduction can only be applied to our setting if the identity of \WORKER['] is revealed (the identity of \FIRM['] is clear in many-to-one markets), as otherwise the algorithm does not know which ``copy'' of \FIRM is blocking with \WORKER. This echoes again the message from our preceding discussion that one should be careful with applying even known reductions within the context of our coarse query model.}

In either case, the algorithm obtains information about the agents' preferences in the form of ``\AG must precede at least one element of $S$ in a particular preference order.'' We write such constraints as \Constraint{\AG}{S} (the agent to whose preference order the constraint applies will be clear from context), and also as \Constraint{\AG}{\AG[']} when $S=\SET{\AG[']}$ is a singleton.
We will consider sets \CONSTRAINTS of such constraints. We say that a ranking \PREFE is \emph{consistent} with \CONSTRAINTS if it satisfies all constraints in \CONSTRAINTS. Throughout, all of our sets \CONSTRAINTS will be such that at least one ranking (namely, the unknown true ranking) is consistent with~\CONSTRAINTS.
When \CONSTRAINTS is a set consisting only of constraints of the form \Constraint{\AG}{\AG[']}, i.e., involving only the relative order of two elements, then \CONSTRAINTS precisely defines a \emph{strict partial order}, and the rankings consistent with \CONSTRAINTS are exactly the \emph{linear extensions} of \CONSTRAINTS.

\section{One-to-One Matchings} \label{sec:one-to-one}
In this section, we focus on the case in which all agents have a quota of $1$, i.e., the case of women and men. We will write $n = n_{\WOMEN}=n_{\MEN}$.
Recall that in this special case, the revelation of a blocking pair $\BLOCKING[t] = (\BWOMAN{t}, \BMAN{t}) = (\WOMAN,\MAN)$ implies that \PrefT{\WOMAN}{\MAN}{\Partner{\WOMAN}} and \PrefT{\MAN}{\WOMAN}{\Partner{\MAN}}; this information reduces the number of (remaining) possible preferences, by reducing either the number of possible preference orders \PREFT{\WOMAN} or the number of possible preference orders \PREFT{\MAN}.

\subsection{High-Level Overview}

The high-level idea of our approach is to choose \MATCH[t] in such a way that whichever blocking pair \BLOCKING[t] the adversary reveals, the number of possible preference orders for at least one of $\WOMAN,\MAN$ not only decreases, but in fact decreases by a
\emph{constant} factor.

For each agent $\AG \in \AGENTS$, let \Constraints[t]{\AG} be the set of all preference constraints that can be inferred from the blocking pairs revealed before (but not including) time $t$.
That is, for a woman \WOMAN, we have $\Constraints[t]{\WOMAN} = \SetB{\Constraint{\MAN}{\Partner[t']{\WOMAN}}}{t' < t \text{ and } \BLOCKING[t'] = \SET{\WOMAN,\MAN}}$; similarly for men.
Let \Rankings[t]{\AG} denote the set of all linear extensions of \Constraints[t]{\AG}.

We call a \emph{scenario} (at time $t$) a vector giving for each woman and each man a preference order consistent with all observations so far; that is, a vector $(\PREF{\AG})_{\AG \in \AGENTS} \in \bigtimes_{\WOMAN \in \WOMEN} \Rankings[t]{\WOMAN} \times \bigtimes_{\MAN \in \MEN} \Rankings[t]{\MAN}$.
Let $\NumRankings[t]{\AG} = \SetCard{\Rankings[t]{\AG}}$ denote the number of remaining possible preference orders for \AG, and $\NUMRANKINGS[t] = \prod_{\AG \in \AGENTS} \NumRankings[t]{\AG}$ the total number of scenarios.
Initially, for every agent \AG, there are $\NumRankings[1]{\AG} = n!$ possible preference orders, so $\NUMRANKINGS[1] = (n!)^{2n}$.

Our algorithm will ensure that the number of possible preference orders for at least one agent, and thus the number of possible scenarios, decreases by a constant factor in each round $t$.
Thus, after at most $t^* = \log \bigl((n!)^{2n}\bigr) = O(n^2 \log n)$ rounds, there is at most one scenario remaining, at which point the algorithm will have identified all of the true preference orders \PREFT{\AG}.
By running the Gale-Shapley Algorithm using those preference orders, the algorithm then ensures that it has found a stable matching.\footnote{It may seem wasteful to pin down the precise preferences/scenario, since many different scenarios may give rise to the same stable matching. Instead, it is tempting to try and find a matching that is stable with respect to all of the remaining possible scenarios (if such a matching exists), and propose that matching rather than further rule out scenarios. The results of \citet{rastegari:condon:immorlica:irving:leyton-brown} imply that checking whether all remaining possible scenarios share the same men-optimal stable matching (i.e., whether the Gale-Shapley Algorithm would return the same matching for all remaining possible scenarios) can be done in polynomial time. However, in such a case, our algorithm would propose this matching as well (as it runs the Gale-Shapley Algorithm on one of the remaining scenarios and proposes the result). It is technically possible to (inefficiently) go over all possible (not necessarily men-optimal) matchings and check whether one of them is stable for all remaining possible scenarios (and then propose it). However, it is not clear how to argue that this would \emph{guarantee} a significantly improved query complexity.}

\subsection{Representative Rankings and Proportional Transitivity}
The basic idea of our algorithm is the same as in the $O(n^3)$ algorithm of Section~\ref{sec:introduction}. In each round, the algorithm chooses a speculative preference order for each agent, and then computes a stable matching with respect to those speculative preference orders. The important improvement is to choose the speculative preference orders for all agents carefully. Specifically, we want all preference orders \PREF[t]{\AG} to be representative of the corresponding \Rankings[t]{\AG}, in the following sense.

\begin{definition}[Representative preference order] \label{def:representative-preference}
Let \CONSTRAINTS be a set of constraints (forming a partial order), and \RANKINGS the set of all linear extensions of \CONSTRAINTS.
Let $\PrefFrac{\AG}{\AG[']} = \PrefFrac[\CONSTRAINTS]{\AG}{\AG[']}$ be the fraction of linear orders \PREFE in \RANKINGS that have \PrefE{\AG}{\AG[']}.
We say that \PREFE is \emph{$\alpha$-representative} of \RANKINGS iff for all pairs $(\AG,\AG['])$, whenever $\PrefFrac{\AG}{\AG[']} \geq \alpha$, we have \PrefE{\AG}{\AG[']}. 
\end{definition}

It is not a priori clear that an $\alpha$-representative preference order even exists for a constant $\alpha < 1$. This is established by the following lemma.

\begin{lemma} \label{lem:representative-exists}
  Let \CONSTRAINTS and $\alpha \geq 0.8$ be arbitrary, and let $\mathcal{S}_{\alpha}(\CONSTRAINTS) = \SetB{(\AG,\AG['])}{\PrefFrac{\AG}{\AG[']} \geq \alpha}$ be the set of all ordered pairs $(\AG,\AG['])$ such that at least an $\alpha$ fraction of the linear extensions of \CONSTRAINTS rank \AG ahead of \AG['].
  Then, $\mathcal{S}_{\alpha}(\CONSTRAINTS)$ is acyclic, and thus defines a partial order. In particular, every linear extension of $\mathcal{S}_{\alpha}(\CONSTRAINTS)$ is an $\alpha$-representative preference order.
\end{lemma}

The proof of Lemma~\ref{lem:representative-exists} relies heavily on Yu's Proportional Transitivity Theorem \citep{yu:proportional-transitivity}.

\begin{theorem}[Theorem~3 of \citealp{yu:proportional-transitivity}]
  \label{thm:proportional-transitivity}
  Let $\phi_{\min} = \frac{1+(\sqrt{2}-1) \cdot \sqrt{2 \sqrt{2}-1}}{2} \approx 0.78$, and $\phi \geq \phi_{\min}$ any constant.
  Let $P$ be any partial order.
  For every $a, a', a''$, if $\PrefFrac[P]{\AG}{\AG[']} \geq \phi$ and $\PrefFrac[P]{\AG[']}{\AG['']} \geq \phi$,
  then $\PrefFrac[P]{\AG}{\AG['']} \geq \phi$.
\end{theorem}

Theorem~\ref{thm:proportional-transitivity} is remarkable: it states that if a large enough fraction of linear extensions of a partial order have \AG ahead of \AG['], and (at least) the same fraction have \AG['] ahead of \AG[''], then (at least) the same fraction have \AG ahead of \AG['']. For $p=1$, this statement is of course trivial (because it exactly captures transitivity of partial orders), but that it should hold for a constant $p < 1$ is not at all obvious.

\begin{proof}[Proof of~\cref{lem:representative-exists}]
  Assume (for contradiction) that $\mathcal{S}_{\alpha}(\CONSTRAINTS)$ contains a cycle $\AG[1] \succ \AG[2] \succ \AG[3] \succ \cdots \succ \AG[\ell] \succ \AG[1]$ of length $\ell$.
  Then, $\PrefFrac[\CONSTRAINTS]{\AG[i]}{\AG[i+1]} \geq \alpha$ for all $i = 1, \ldots, \ell-1$.
  Applying Theorem~\ref{thm:proportional-transitivity} with $\phi = \alpha \geq 0.8 > \phi_{\min}$ repeatedly, we obtain that $\PrefFrac[\CONSTRAINTS]{\AG[1]}{\AG[\ell]} \geq \alpha$.
  But this gives a contradiction with the fact that $\PrefFrac[\CONSTRAINTS]{\AG[\ell]}{\AG[1]} \geq \alpha$, because $\PrefFrac[\CONSTRAINTS]{\AG[1]}{\AG[\ell]} + \PrefFrac[\CONSTRAINTS]{\AG[\ell]}{\AG[1]} = 1$.
\end{proof}

\subsection{The Algorithm and Analysis}

The high-level algorithm is given in \cref{algo:learning-matching}.
In each iteration, it uses $\alpha$-representative preference orders (for some absolute constant $\alpha \in [0.8,1)$) as speculative preference orders for all agents~\AG;
by Lemma~\ref{lem:representative-exists}, such preference orders actually exist.
It then uses the Gale-Shapley Algorithm (or really, any algorithm) to compute a matching that is stable with respect to those speculative preferences.
If this matching is not stable, then it updates the constraints based on the blocking pair that was revealed.

\begin{algorithm}[htb]
\begin{algorithmic}[1]
\STATE Let $t=1$.
\FORALL{$\AG\in\AGENTS$}
  \STATE Let $\Constraints[1]{\AG} = \emptyset$.
\ENDFOR
\LOOP
  \FORALL{$\AG\in\AGENTS$}
     \STATE Let \PREF[t]{\AG} be an $\alpha$-representative preference order for \Constraints[t]{\AG}.\label[line]{choose-representative-one-to-one}
  \ENDFOR
  \STATE Let $\MATCH[t] = \GaleShapley{(\PREF[t]{\AG})_{\AG \in \AGENTS}}$ be a stable matching with respect to the \PREF[t]{\AG}.
  \IF{\MATCH[t] is stable}
    \RETURN \MATCH[t].
  \ELSE
    \STATE Let $\SET{\WOMAN,\MAN}=\BLOCKING[t]$  be the revealed blocking pair.
    \STATE Let $\Constraints[t+1]{\WOMAN} = \Constraints[t]{\WOMAN} \cup \SETB{\Constraint{\MAN}{\Partner[t]{\WOMAN}}}$.
    \STATE Let $\Constraints[t+1]{\MAN} = \Constraints[t]{\MAN} \cup \SETB{\Constraint{\WOMAN}{\Partner[t]{\MAN}}}$.
    \FORALL{$\AG\notin\SET{\WOMAN,\MAN}$}
      \STATE Let $\Constraints[t+1]{\AG} = \Constraints[t]{\AG}$.
    \ENDFOR
  \ENDIF
  \STATE Increment $t$.
\ENDLOOP
\end{algorithmic}

\caption{The interactive learning algorithm for one-to-one matchings. \label{algo:learning-matching}}
\end{algorithm}

The total number of iterations required by \cref{algo:learning-matching} is characterized by the following theorem.

\begin{theorem}\label{thm:learning-matching-iterations}\!\!\!\footnote{In proof: after this paper appeared at EC'20, the authors of \citep{bei:chen:zhang:complexity} contacted us (personal communication, August 20, 2020) and made us aware of their paper, of which we had been unaware until that point. As they pointed out, their paper contains a proof of \cref{thm:learning-matching-iterations}. While of the same order, the number of queries that our algorithm uses is smaller by a constant factor, and our proof is somewhat shorter, since we used \citeauthor{yu:proportional-transitivity}'s Proportional Transitivity Theorem, whereas they proved the result from first principles; however, neither this difference nor the difference in proofs is substantial. See the newly added \cref{bei-chen-zhang} following our introduction for more details.}
  Algorithm~\ref{algo:learning-matching} terminates after at most $O(n^2 \log n)$ iterations.
\end{theorem}

\begin{proof}
  Because $\alpha \geq 0.8$, by Lemma~\ref{lem:representative-exists}, the algorithm can always find representative preference orders to use as speculative preference orders.
  
  Consider any iteration $t$ in which a blocking pair $\BLOCKING[t] = \SET{\WOMAN,\MAN}$ (with respect to the true preferences \PREFT{\AG}) was revealed.
  Because the Gale-Shapley Algorithm finds a stable matching, the pair $\SET{\WOMAN,\MAN}$ was not blocking for the preferences \PREF[t]{\AG} used in running the Gale-Shapley Algorithm. In particular, this means that at least one of $\WOMAN,\MAN$ preferred their assigned partner over the other, i.e., \Pref[t]{\WOMAN}{\Partner{\WOMAN}}{\MAN} or \Pref[t]{\MAN}{\Partner{\MAN}}{\WOMAN}. Without loss of generality, assume that \Pref[t]{\WOMAN}{\Partner{\WOMAN}}{\MAN}.

  Because \PREF[t]{\WOMAN} was $\alpha$-representative, $\PrefFrac[{\Constraints[t]{\WOMAN}}]{\Partner{\WOMAN}}{\MAN} \geq 1-\alpha$. The revelation of the blocking pair $\SET{\WOMAN,\MAN}$ rules out all linear orders \PREFT{\WOMAN} with \PrefT{\WOMAN}{\Partner{\WOMAN}}{\MAN}; we therefore obtain that $\NumRankings[t+1]{\WOMAN} \leq \alpha \cdot \NumRankings[t]{\WOMAN}$. In turn, this implies that $\NUMRANKINGS[t+1] \leq \alpha \cdot \NUMRANKINGS[t]$.

  We obtain that Algorithm~\ref{algo:learning-matching} terminates after at most $T = \log_{\nicefrac{1}{\alpha}} (n!)^n = O(\frac{\alpha}{1-\alpha} \cdot n^2 \log n)=O(n^2\log n)$ iterations.
\end{proof}

\begin{remark}
The number of iterations is $O(\frac{\alpha}{1-\alpha} \cdot n^2 \log n)$, but because $\alpha$ is an absolute constant (e.g., $\alpha = 0.8$, or $\alpha = 0.9$ in the efficient implementation to follow), the number of iterations is simply $O(n^2 \log n)$.
\end{remark}

The implementation of \cref{algo:learning-matching}, as given above, is not efficient, because we did not specify how to find representative preference orders. One way to do so is to enumerate all linear extensions of \Constraints[t]{\AG}, and use this enumeration to compute the set $\mathcal{S}_{\alpha}(\Constraints[t]{\AG}) = \SetB{(\AG,\AG['])}{\PrefFrac{\AG}{\AG[']} \geq \alpha}$ from Lemma~\ref{lem:representative-exists} explicitly. Given $\mathcal{S}_{\alpha}(\Constraints[t]{\AG})$, finding a linear extension is simply a topological sort. Of course, the enumeration may take time $\Omega(n!)$. In the next section, we give an efficient randomized implementation, which terminates in $O(n^2 \log n)$ iterations with high probability.

\subsection{An Efficient Randomized Implementation}
In this section, we describe and analyze a randomized algorithm for sampling representative preference orders for a set of constraints (partial order) \CONSTRAINTS.
The algorithm is given as Algorithm~\ref{algo:sample-preference}, and relies heavily on the following theorem of \citet{huber:linear-extensions}, which is based on the Karzanov-Khachiyan Chain \citep{karzanov:khachiyan:conductance} and a bounding chain:

\begin{theorem}[Lemma~10 of \citealp{huber:linear-extensions}]
\label{thm:huber-theorem}
There exists a randomized algorithm for sampling linear extensions of a partial order \CONSTRAINTS, with the following guarantees:
\begin{enumerate}
\item The algorithm samples linear extensions exactly uniformly from all linear extensions of \CONSTRAINTS.
\item The expected running time of the algorithm is at most $4.3 \cdot n^3 \ln n$.
\item For every integer $k$, with probability at least $1\!-\!\nicefrac{1}{n^{k-1}}$, the running time of the algorithm is at most $k \cdot \frac{16}{\pi^2} \cdot n^3 \ln n$.
\end{enumerate}
\end{theorem}

\begin{algorithm}[htb]
\begin{algorithmic}[1]
  \STATE Use Huber's Algorithm (see \cref{thm:huber-theorem}) to independently uniformly sample $K=600 \ln(n)$ linear extensions $\PREF{1}, \PREF{2}, \ldots, \PREF{K}$ of \CONSTRAINTS.
  \FORALL{$(\AG,\AG['])\in A^2$ s.t.\ $\AG\ne\AG[']$}
    \STATE Let \PrefFracE{\AG}{\AG[']} be the fraction of extensions \PREF{i} that have \Pref{i}{\AG}{\AG[']}.
  \ENDFOR
  \STATE Let $\SamplePO = \SetB{\Constraint{\AG}{\AG[']}}{\PrefFracE{\AG}{\AG[']} \geq 0.85}$.
  \IF{\SamplePO is acyclic}
    \STATE Use Topological Sort to compute an arbitrary linear extension \PREFH of \SamplePO.
    \STATE Return \PREFH.
  \ELSE
    \STATE Start over from the beginning.
  \ENDIF
\end{algorithmic}

\caption{The algorithm $\SampleRanking(\CONSTRAINTS)$. \label{algo:sample-preference}}
\end{algorithm}

The two key observations about Algorithm~\ref{algo:sample-preference} are  the following: (1) the set \SamplePO of sampled constraints is acyclic with high probability, ensuring (fast) termination of the algorithm; and (2) with high probability, \PREFH is 0.9-representative.

To prove both statements, we begin with a simple tail bound, capturing that the \PrefFracE{\AG}{\AG[']} are accurate estimates of \PrefFrac{\AG}{\AG[']}.
Let \Event{E} denote the event that all estimates \PrefFracE{\AG}{\AG[']} simultaneously satisfy that $\Abs{\PrefFracE{\AG}{\AG[']} - \PrefFrac[\CONSTRAINTS]{\AG}{\AG[']}} \leq 0.05$.

\begin{lemma} \label{lem:good-estimates}
  $\Prob{\Event{E}} \geq 1-\frac{1}{n}$.
\end{lemma}

\begin{proof}
  First focus on one pair $(\AG,\AG['])$.
  By \cref{thm:huber-theorem}, the linear extensions are sampled
  \emph{exactly} from the uniform distribution.\footnote{This is only needed in our proof in as much as it keeps the analysis clean. The Markov chain of \citet{karzanov:khachiyan:conductance} and the improved version of \citet{bubley:dyer:linear-extensions} ensure rapid mixing; by running the chains long enough, one can easily ensure that the sampling error is small enough to not affect our analysis.}
  Therefore,
  we have $\Expect{\PrefFracE{\AG}{\AG[']}} = \PrefFrac[\CONSTRAINTS]{\AG}{\AG[']}$.
  The Hoeffding Bound\footnote{$\ProbB{\AbsB{\bar{X} - \Expect{\smash{\bar{X}}}} \geq t} \leq 2 \e^{-2nt^2}$ when $\bar{X} = \nicefrac{1}{n} \cdot \sum_{i=1}^n X_i$ and the $X_i$ are independent random variables supported on $[0,1]$.}
  \citep{hoeffding} guarantees that
  $\ProbB{\Abs{\PrefFracE{\AG}{\AG[']}-\PrefFrac[\CONSTRAINTS]{\AG}{\AG[']}} \geq 0.05} \leq 2\exp(-\nicefrac{K}{200}) = 2n^{-3}$.
  A union bound over all $\binom{n}{2}$ pairs $\AG,\AG[']$ now implies that all estimates are accurate to within an additive 0.05 with probability at least~$1\!-\!\frac{1}{n}$.
\end{proof}

\begin{lemma} \label{lem:acyclic-constraints}
   Under \Event{E}, the set \SamplePO of preference constraints is acyclic.
\end{lemma}

\begin{proof}
  Because we are conditioning on the event \Event{E}, for all constraints $\Constraint{\AG}{\AG[']} \in \SamplePO$, we must have that $\PrefFrac[\CONSTRAINTS]{\AG}{\AG[']} \geq 0.8$.
  Acyclicity now follows from \cref{lem:representative-exists}, because $\SamplePO \subseteq \mathcal{S}_{0.8}(\CONSTRAINTS)$ for the set $\mathcal{S}_{0.8}(\CONSTRAINTS)$ defined in that \lcnamecref{lem:representative-exists}.
\end{proof}

\begin{lemma} \label{lem:large-fraction}
  Under \Event{E}, the preference order \PREFH returned by Algorithm~\ref{algo:sample-preference} is 0.9-representative.
\end{lemma}

\begin{proof}
  Let $\AG, \AG[']$ be such that $\PrefFrac[\CONSTRAINTS]{\AG}{\AG[']} \geq 0.9$. Because we are conditioning on \Event{E}, this implies that $\PrefFracE{\AG}{\AG[']} \geq 0.85$, so by definition, $\Constraint{\AG}{\AG[']} \in \SamplePO$. Because \PREFH is a linear extension of \SamplePO, we get that \PrefH{\AG}{\AG[']}.
\end{proof}

We now return to the analysis of \cref{algo:learning-matching}, proving the following theorem:

\begin{theorem}\label{thm:learning-matching-time}\!\!\!\footnote{In proof: after this paper appeared at EC'20, the authors of \citep{bei:chen:zhang:complexity} contacted us (personal communication, August 20, 2020) and made us aware of their paper, of which we had been unaware until that point. As they pointed out, their paper contains a proof of \cref{thm:learning-matching-time}. While of the same order, the number of queries that our algorithm uses is smaller by a constant factor, and our proof is somewhat shorter, since we used \citeauthor{yu:proportional-transitivity}'s Proportional Transitivity Theorem, whereas they proved the result from first principles; however, neither this difference nor the difference in proofs is substantial. See the newly added \cref{bei-chen-zhang} following our introduction for more details.}
  If $\SampleRanking(\Constraints[t]{\AG})$ is used to compute the speculative preference order \PREF[t]{\AG} in \cref{choose-representative-one-to-one} of \cref{algo:learning-matching}, then \cref{algo:learning-matching} terminates after at most $O(n^2 \log n)$ iterations with high probability; furthermore, with probability at least $1\!-\!O(\nicefrac{1}{n^2})$, each iteration runs in time $O(n^3 \log^2 n)$.
\end{theorem}

\begin{proof}
Consider running \cref{algo:learning-matching} for $\Theta(n^2 \log n)$ iterations, with a suitably large constant.
Let~$N$ be the number of those iterations in which the event \Event{E} held.
By Lemma~\ref{lem:good-estimates}, $N$ is lower-bounded by a $\mathrm{Binomial}\bigl(\Theta(n^2 \log n),1\!-\!\frac{1}{n}\bigr)$ random variable. By Chernoff Bounds, with high probability, $N = \Theta(n^2 \log n)$, again for a suitably large constant.

During any iteration $t$ in which \Event{E} did \emph{not} hold, we simply (pessimistically) bound $\NUMRANKINGS[t+1] \leq \NUMRANKINGS[t]$. For the remaining iterations, Lemma~\ref{lem:large-fraction} guarantees that the \PREF[t]{\AG} are 0.9 representative. Therefore, using $\alpha=0.9$, Theorem~\ref{thm:learning-matching-iterations} guarantees that when $N = \Omega(n^2 \log n)$, the algorithm terminates with a stable matching.

Finally, we analyze the running time.
The running time of \cref{algo:sample-preference} is dominated by drawing $O(\log n)$ independent uniformly random linear extensions of \CONSTRAINTS.
By \cref{thm:huber-theorem} with $k=6$,
each draw takes time $O(n^3 \log n)$
with probability at least $1\!-\!\nicefrac{1}{n^5}$.
The failure probability of event \Event{E} is at most $\frac{1}{n}$, so with probability at least $1\!-\!\nicefrac{1}{n^5}$, the number of times \cref{algo:sample-preference} is repeated before succeeding is at most 5, and in particular constant. Next, we take a union bound over the following events:
\ \ (1) each of the $O(\log n)$ calls to Huber's algorithm does not take too long, and
\ \ (2) \cref{algo:sample-preference} is not repeated more than $5$ times.
By this union bound, \cref{algo:sample-preference} runs in time $O(n^3 \log^2 n)$ with probability at least $1\!-\!O(\nicefrac{\log n}{n^5})$.

In each iteration $t$, \cref{algo:learning-matching} in Line~\ref{choose-representative-one-to-one} needs to ``construct'' $\alpha$-representative preference orders for \Constraints[t]{\AG} for all agents \AG.
In the first iteration, this is trivial, as any order is $\alpha$-representative.
In subsequent iterations, only the orders for the two agents who were part of blocking pairs may have to be updated.
Thus, each iteration requires two calls to \cref{algo:sample-preference}, for a total time of $O(n^3 \log^2 n)$.
By a union bound over these two invocations of \cref{algo:sample-preference}, each iteration of \cref{algo:learning-matching} takes time $O(n^3 \log^2 n)$ in calls to \cref{algo:sample-preference}, with probability at least $1\!-\!O(\nicefrac{\log n}{n^5})$.
This (high-probability) running time of $O(n^3 \log^2 n)$ dominates the running time of the Gale-Shapley Algorithm and of updating the constraint sets in each iteration.

  Finally, by a union bound over all $O(n^2 \log n)$ iterations of \cref{algo:learning-matching}, the running time guarantee holds with probability at least $1\!-\!O(\nicefrac{\log^2 n}{n^3}) = 1\!-\!O(\nicefrac{1}{n^2})$.
\end{proof}

\begin{remark}
The guarantee of $1\!-\!O(\nicefrac{1}{n^2})$ in \cref{thm:learning-matching-time} can be improved to an arbitrary polynomial, by choosing a larger constant in the application of Huber's Theorem, then following the same reasoning and union bounds.
\end{remark}

\subsection{A Near-Matching Lower Bound}
We next prove a lower bound on the number of rounds required, which matches the upper bound up to a factor of $\Theta(\log n)$.

\begin{proposition}\label{prop:lower-bound}\!\!\!\footnote{In proof: after this paper appeared at EC'20, the authors of \citep{bei:chen:zhang:complexity} contacted us (personal communication, August 20, 2020) and made us aware of their paper, of which we had been unaware until that point. As they pointed out, their paper contains a theorem similar to \cref{prop:lower-bound}. Our result is somewhat stronger by allowing preferences on one side to be known and identical, and preferences on the other side to be i.i.d.\ random, whereas theirs is a worst-case result; however, the qualitative message is similar. See the newly added \cref{bei-chen-zhang} following our introduction for more details.}
Let $\mathcal{A}$ be any interactive algorithm for finding a stable matching, randomized or deterministic, and with arbitrary computational requirements.
Even when\ \ (1) the preferences of all women are known to $\mathcal{A}$ and identical, and\ \ (2) the preferences of all men are i.i.d.\ uniformly random, against adversarial answers, $\mathcal{A}$ must take at least $\Omega(n^2)$ rounds in expectation.
\end{proposition}

\begin{proof}
  The high-level idea of the proof is quite straightforward, but a formal proof needs some care in the details. We begin with the very straightforward case of adversarially chosen preferences, and then discuss the intuition of the proof for random preferences.
  Let $\MAN[1] \succ \MAN[2] \succ \cdots \succ \MAN[n]$ be the men ranked by the common (and known) preference of all women.
  For each man $i$, let \PREFT{i} be his preference order over women. As is well known, the unique stable matching when women have common preferences is the result of a \emph{serial dictatorship} of the men in the order of these preferences. That is, \MAN[1] is matched with his first choice, \MAN[2] is matched with his first choice among the remaining women, and so on.

  An adversary can use the following strategy. Whatever matching the algorithm proposes in the first iteration, the woman \WOMAN matched to \MAN[1] will be adversarially placed as his last-ranked choice, and the algorithm will reveal a blocking pair of \MAN[1] with some other woman, and will continue revealing that same blocking pair whenever \MAN[1] is matched with \WOMAN. Whatever different woman \MAN[1] is next matched with (whether it is the one that was revealed or not) will be adversarially placed as \MAN[1]'s second-to-last choice, and so on. Once the algorithm has matched \MAN[1] with all $n$ women, the the last woman to be matched with him is known to be his first choice. At this point, the adversary will continue with the same strategy for \MAN[2] with the remaining $n\!-\!1$ women. Because only blocking pairs involving \MAN[1] have been revealed up to that point, the ranking of \MAN[2] is completely undetermined. The adversary continues in this way for all men in order. In total, this adversary strategy will force any algorithm to use at least $\sum_{i=0}^{n-1} i = \Omega(n^2)$ rounds of queries.

  When the men have i.i.d.\ uniformly random preference rankings, the high-level idea remains the same. For any $i$, until the algorithm has identified the correct matches of $\MAN[1], \ldots, \MAN[i-1]$, the adversary will not reveal any blocking pair involving $\MAN[i], \ldots, \MAN[n]$.
However, in contrast with the case of adversarial rankings, once all of these $i\!-\!1$ matches have been identified,  the adversary cannot ensure any more that the woman \WOMAN with whom \MAN[i] is matched is the lowest remaining in his ranking. Nevertheless, the adversary can still reveal a block with the woman immediately preceding \WOMAN in \MAN[i]'s ranking.
 Intuitively (we will make this argument precise, of course), this ``does not reveal much information'' besides the fact that \WOMAN is not \MAN[i]'s correct match, so the algorithm in expectation still has to try a constant fraction of all remaining women as \MAN[i]'s match.

To precisely define the adversary's behavior in the random-preferences case, fix realized preferences for all men. We inductively define the following women and  sets of women. \RSET{1} is the set of all women. For each $i$, let \RWOMAN[i] be the top choice of \MAN[i] among the women in \RSET{i}, according to the uniformly random (but fixed for now) preference ranking \PREFT{i}. Then, let $\RSET{i+1} = \RSET{i} \setminus \SET{\RWOMAN[i]}$. The unique stable matching has each \MAN[i] matched with \RWOMAN[i].
The adversary behaves as follows: if the proposed matching is not stable, then let $i = \min \SetB{j}{\Partner[t]{\MAN[j]} \neq \RWOMAN[j]}$. Reveal the blocking pair consisting of \MAN[i] and \WOMAN['], where the latter is the \emph{immediate predecessor} of $\WOMAN=\Partner[t]{\MAN[i]}$ among the set \RSET{i} according to \PREFT{\MAN[i]}.

  Making the preceding informal argument --- about this adversary not revealing ``much information'' --- precise requires some care. In particular, we must reason about what the algorithm has ``learned'' about the uniformly random preferences of the men. Our analysis will therefore use the Principle of Deferred Decisions, and consider the random preferences of the men as if they are drawn in stages over the course of the algorithm's execution. However, these deferred random choices will of course be sound, as we will show that the final preferences are i.i.d.~uniformly random, regardless of the matchings proposed by the learning algorithm $\mathcal{A}$. Due to the probabilistic nature of our analysis, \RWOMAN[i] and \RSET{i} will henceforth be treated as random variables.

  In each round $t$, the algorithm \ALG proposes a matching \MATCH[t]. The adversary reveals a blocking pair involving \MAN[i] if and only if $i = \min \SetB{j}{\Partner[t]{\MAN[j]} \neq \RWOMAN[j]}$. Naturally, an adversary following this strategy will reveal to the algorithm that $\Partner[t]{\MAN[j]} = \RWOMAN[j]$ for all $j < i$. Because blocking pairs involving \MAN[i] are only revealed in this case, we obtain the following first key property:\ \ (1) each blocking pair for \MAN[i] involves $\Partner[t]{\MAN[i]} \in \RSET{i}$ and some other $\WOMAN \in \RSET{i}$.
Since the adversary will always reveal a \WOMAN who is the \emph{immediate} predecessor of \Partner[i]{\MAN[i]} among \RSET{i},
this implies a second key property:\ \ (2) \Constraints[t]{m_i}, which to avoid clutter we denote by \Constraints[t]{i}, defines a union of disjoint \emph{chains} on~\RSET{i}.

  In our analysis, we consider as already drawn the following information about all of the \PREFT{i} at all times $t$, which contains all of the information known to the algorithm:\ \ (1) The preference constraints \Constraints[t]{i} (the blocking pairs involving \MAN[i] already revealed to the algorithm), and\ \ (2) possibly the realization \WOMAN[i] of the random \RWOMAN[i], if it has been revealed, in which case we consider the entire preference relation \PREFT{i} as having been drawn already. Specifically, this information is tracked using the variables $\Constraints{i}$ and $w_i$, where initially we have $\Constraints{i}=\emptyset$ and $\WOMAN[i] = \perp$ for all $i$. Whenever the adversary commits to a blocking pair, we update the variable \Constraints{i}, and whenever the adversary commits to the value of \RWOMAN[i] in round $t$, we update the variable $w_i$ to reflect the realized value of \RWOMAN[i] and update the variable $\Constraints{i}$ to reflect the realization of all of \PREFT{i}.
  We write $\CHAINS[t]{i} = \SetCard{\RSET{i}} - \SetCardB{\Constraints[t]{i}} = n+1-i - \SetCardB{\Constraints[t]{i}}$ for the number of chains in the partial order defined by \Constraints[t]{i}.
  The adversary's strategy alongside the gradual draw of preferences, as viewed from the Deferred Decision point of view of our analysis, is given in Algorithm~\ref{algo:adversary}.

\begin{algorithm}[htb]
\begin{algorithmic}[1]
\STATE Let \MATCH[t] be the proposed matching.
\STATE Let $i_{\perp}$ be minimal such that $w_{i_{\perp}}=\perp$, or $n+1$ if $w_i\ne\perp$ for every $i\in\{1,\ldots,n\}$.
     \IF{there exists an $i < i_{\perp}$ such that $\Partner[t]{\MAN[i]}\ne\WOMAN[i]$}
       \RETURN the blocking pair $(\MAN[i], \WOMAN['])$ for the minimal such $i$ and for \WOMAN['] that immediately precedes \Partner[t]{\MAN[i]} according to \Constraints{i}.
       \COMMENT{Reveals no new information regarding men for whom $w_i=\perp$.}
     \ENDIF
     \IF{\Partner[t]{\MAN[i_{\perp}]} has a predecessor \WOMAN['] in \Constraints{i_{\perp}}}
       \RETURN the blocking pair $(\MAN[i_{\perp}], \WOMAN['])$.
       \COMMENT{Reveals no new information.}
     \ENDIF
     \FOR{$i=i_{\perp},\ldots,n$}
       \STATE Flip a coin w.p.\ $\nicefrac{1}{\CHAINS[t]{i}}$ of coming up heads. \label{alg:adversary:assignment} \COMMENT{Equals $\ProbCB{\Partner[t]{\MAN[i]}=\RWOMAN[1]}{w_1,\ldots,w_{i-1},\Constraints{i}}$}
       \IF[The algorithm chose an unstable partner for {\MAN[i]}.]{the coin comes up tails}
         \STATE Let $\WOMAN['] \in \RSET{i}$ be uniformly random among all women such that there exists no \WOMAN[''] with $\Constraint{\WOMAN[']}{\WOMAN['']} \in \Constraints{i}$, i.e., all women who have not been revealed as direct predecessors (in \RSET{i}) of other women. \label{alg:adversary:predecessor-choice}
         \STATE Add the constraint \Constraint{\WOMAN[']}{\Partner[t]{\MAN[i]}} to \Constraints{i}.
         \RETURN the blocking pair $(\MAN[i], \WOMAN['])$.
       \ELSE[The algorithm chose the stable partner for {\MAN[i]}.]
         \STATE Set $\WOMAN[i] = \Partner[t]{\MAN[i]}$.
         \STATE Let \PREFT{i} be a uniformly random linear extension of \Constraints{i}, subject to the constraint that \WOMAN[i] precedes all elements of \RSET{i}.
         \STATE Update \Constraints{i} to capture all of \PREFT{i}.
         \COMMENT{The next iteration of the \textbf{for} loop will now commence.}
       \ENDIF
\ENDFOR
\LINECOMMENT{If this point has been reached, each man \MAN[i] must be matched with his stable partner \WOMAN[i].}
\RETURN ``The matching is stable.''
\end{algorithmic}

\caption{The adversary's strategy for revealing a blocking pair, combined with the gradual draw of preferences, from the Deferred Decisions point of view of our analysis. \label{algo:adversary}}
\end{algorithm}

We first verify that \cref{algo:adversary}, using deferred decisions, generates uniformly random rankings regardless of the proposed matchings \MATCH[t].
To see that this, recall the key properties (1) and (2) above, which state that the revealed constraints \Constraints[t]{i} partition \RSET{i} into disjoint chains of immediate predecessors. In a uniformly random permutation, the order of these chains is uniformly random.
In particular, when there are \CHAINS[t]{i} chains, each woman without a known predecessor is first in the ranking with probability $\nicefrac{1}{\CHAINS[t]{i}}$.
The bias of the random coin flip in Line~\ref{alg:adversary:assignment} guarantees that when any woman \WOMAN without a known predecessor is considered, she is revealed to be first in the ranking with exactly this probability $\nicefrac{1}{\CHAINS[t]{i}}$. 
Conditioned on not being first overall in the ranking, \WOMAN has a predecessor who is uniformly random among the last elements of all other $\CHAINS[t]{i}-1$ chains.
The draw of \WOMAN['] in Line~\ref{alg:adversary:predecessor-choice} exactly matches this distribution.

Finally, we analyze the number of iterations that \ALG needs. In Step~\ref{alg:adversary:assignment} of Algorithm~\ref{algo:adversary}, the probability that \WOMAN[i] is revealed (and thus the algorithm can move on to finding \MAN[i+1]) is $\frac{1}{\CHAINS[t]{i}} = \frac{1}{n+1-i - \SetCard{\Constraints[t]{i}}}$.
For $i \leq \nicefrac{n}{3}$ and the first $\nicefrac{n}{3}$ iterations for \MAN[i] (where $\Constraints[t]{i} \leq \nicefrac{n}{3}$), we obtain that the probability of finding \WOMAN[i] is at most $\nicefrac{3}{n}$. Thus, in expectation, it takes at least $\nicefrac{n}{3}$ iterations each to identify $\WOMAN[1], \ldots, \WOMAN[n/3]$, implying that \ALG needs at least $\nicefrac{n^2}{9}$ iterations in expectation.
\end{proof}

\subsection{Many-to-Many Matchings with Extra Information}\label{many-many-with-discard-information}
Our algorithms and analysis would have remained essentially unchanged for the case of many-to-many matchings had the algorithm, along with a blocking pair $\BLOCKING = \SET{\WORKER,\FIRM}$, learned $\FIRM['] \in \Partners{\WORKER}, \WORKER['] \in \Partners{\FIRM}$ such that \PrefT{\WORKER}{\FIRM}{\FIRM[']} and \PrefT{\FIRM}{\WORKER}{\WORKER[']}. In this case, the extra information would have still allowed the algorithm to add at least one new constraint \Constraint{\FIRM}{\FIRM[']} or \Constraint{\WORKER}{\WORKER[']} to \Constraints{\WORKER} or \Constraints{\FIRM} in each iteration. So Algorithm $\SampleRanking$ (\cref{algo:sample-preference}) could have been used entirely unchanged, and the standard many-to-many variant of the Gale-Shapley Algorithm could have been used to compute a stable matching with respect to the sampled preference orders.

Computing a stable many-to-many matching with respect to unknown preferences becomes significantly more challenging when the algorithm only sees the blocking pair $(\WORKER,\FIRM)$, with no additional information about which agent in \Partners{\WORKER} or \Partners{\FIRM} should be ``swapped out.''
The modified algorithms and analysis for this case are the subject of Section~\ref{sec:many-to-many}.

\section{Many-to-Many Matchings} \label{sec:many-to-many}
In this section, we describe and analyze algorithms for many-to-many matchings when the response to a query reveals only a blocking pair $(\WORKER,\FIRM)$, but not the current matches $\FIRM['] \in \Partners{\WORKER}$ and $\WORKER['] \in \Partners{\FIRM}$ such that \PrefT{\WORKER}{\FIRM}{\FIRM[']} and \PrefT{\FIRM}{\WORKER}{\WORKER[']}.

We obtain two results. First, in Section~\ref{sec:small-capacities}, we assume that all quotas are bounded, and write $\MaxCapacity = \max_{\AG} \Capacity{\AG}$.
We give a simple learning algorithm that takes at most $O(n^{\MaxCapacity+2})$ rounds of interaction, with computation time at most $O(n^{\MaxCapacity+3})$ per round --- this algorithm is a conceptual generalization of the simple $O(n^3)$ algorithm for one-to-one matchings discussed in Section~\ref{sec:introduction}.
Then, in Section~\ref{sec:general-capacities}, we present a more involved algorithm that only requires $O(n^3 \log n)$ rounds of interaction, regardless of the quotas. While the number of interactions is always polynomial, the computational requirements are not: we are currently not aware of any implementation that is not exponential in $n$. We show that being able to efficiently sample preference orders subject to constraints more involved than pairwise comparisons would be sufficient to obtain an efficient implementation.

The algorithms for both cases are based on modifying Algorithm~\ref{algo:learning-matching}. The main difference is the type of information that the algorithm obtains from a blocking pair. Whereas previously, a blocking pair $\SET{\WORKER,\FIRM}$ revealed that \WORKER prefers \FIRM over \Partner{\WORKER}, it now only reveals that \WORKER prefers \FIRM over (at least) one participant in \Partners{\WORKER}. 

For the modified algorithm, \Constraints[t]{\AG} will be a set of constraints of the form \Constraint{\AG[']}{S}; in these constraints, all sets $S$ will have the same size, namely, \Capacity{\AG}. The algorithm uses speculative preference orders \PREF[t]{\AG} that are ``representative'' of \Constraints[t]{\AG} in a sense that we will describe below.
The modified learning algorithm is given by Algorithm~\ref{algo:learning-many-to-many-matching}.

\begin{algorithm}[htb]
\begin{algorithmic}[1]
\setcounter{ALC@unique}{0} 
\STATE Let $t=1$.
\FORALL{$\AG\in\AGENTS$}
  \STATE Let $\Constraints[1]{\AG} = \emptyset$.
\ENDFOR
\LOOP
  \FOR{each \AG}
     \STATE Let \PREF[t]{\AG}
     be a ``representative'' preference order for \Constraints[t]{\AG}.\label[line]{choose-representative-many-to-many}
  \ENDFOR
  \STATE Let $\GMATCH[t] = (\Partners{\AG})_{\AG \in \AGENTS} = \GaleShapley{(\PREF[t]{\AG})_{\AG \in \AGENTS}}$ be a stable many-to-many matching with respect to the \PREF[t]{\AG}.
  \IF{\GMATCH[t] is stable}
    \RETURN \GMATCH[t].
  \ELSE
    \STATE Let $\SET{\WORKER,\FIRM}=\BLOCKING[t]$  be the revealed blocking pair.
    \STATE Let $\Constraints[t+1]{\WORKER} = \Constraints[t]{\WORKER} \cup \SETB{\Constraint{\FIRM}{\Partners[t]{\WORKER}}}$.
    \STATE Let $\Constraints[t+1]{\FIRM} = \Constraints[t]{\FIRM} \cup \SET{\Constraint{\WORKER}{\Partners[t]{\FIRM}}}$.
    \FORALL{$\AG\notin\SET{\WORKER,\FIRM}$}
      \STATE Let $\Constraints[t+1]{\AG} = \Constraints[t]{\AG}$.
    \ENDFOR
  \ENDIF
  \STATE Increment $t$.
\ENDLOOP
\end{algorithmic}

\caption{The interactive learning algorithm for many-to-many matchings. \label{algo:learning-many-to-many-matching}}
\end{algorithm}

The difference between the analyses in Sections~\ref{sec:small-capacities} and \ref{sec:general-capacities} is solely in the choice of the speculative preferences, and the corresponding guarantees that can be derived on the progress of Algorithm~\ref{algo:learning-many-to-many-matching} and the running time of the implementation. Specifically, in \cref{sec:small-capacities} we will allow the speculative preferences to be arbitrary preferences that are consistent with \Constraints[t]{\AG}, while in \cref{sec:general-capacities} we will use speculative preferences that are representative of \Constraints[t]{\AG} with respect to a notion of representativeness that will be defined iteratively in terms of fractions of being ranked last.

\subsection{Small Quotas} \label{sec:small-capacities}

The simplest way to implement \cref{choose-representative-many-to-many} in \cref{algo:learning-many-to-many-matching} is to let \PREF[t]{\AG} be \emph{any} preference order consistent with all the constraints in \Constraints[t]{\AG}. We show that such a preference order can be found efficiently.

\begin{proposition} \label{prop:find-any-ranking}
Assume that there exists at least one linear order \PREFH consistent with \CONSTRAINTS. Then, such an order can be found in time $O(n \cdot \SetCard{\CONSTRAINTS})$.
\end{proposition}

\begin{proof}
  The algorithm is a generalization of the standard Topological Sort algorithm for directed acyclic graphs.
  It scans through all the constraints and counts, for each relevant \AG (either $\AG \in \WORKERS$ or $\AG \in \FIRMS$, depending on whether a preference order is being constructed for some $\WORKER \in \WORKERS$ or some $\FIRM \in \FIRMS$), how many constraints there are of the form \Constraint{\AG}{S}. Because each constraint \Constraint{\AG}{S} implies that \AG cannot be last in the preference order, and we assumed that there is a linear order \PREFH consistent with \CONSTRAINTS, there must be at least one \AG for which there are no constraints \Constraint{\AG}{S}.

  The algorithm picks any such \AG and puts it in the last remaining place of the preference order. It then removes all constraints of the form \Constraint{\AG[']}{S} with $S \ni \AG$ from \CONSTRAINTS, because they have been satisfied by placing \AG last. It then continues with the remaining constraints for the remaining positions.

  The algorithm takes $n$ iterations, each of which take time $O(\CONSTRAINTS)$ for the scan and removal.
\end{proof}

While the generalization of Topological Sort to the more general types of constraints is fairly straightforward, we remark on a subtlety that does not arise for standard Topological Sort.
A slight generalization of the types of constraints we consider is to allow constraints both of the form ``\AG must precede at least one element of $S$'' and ``\AG must follow at least one element of $S$.'' When $S$ is a singleton, these types of constraints are the same. When the sets $S$ can be larger, the types of constraints are different, and we showed that when only one type of constraint is allowed, a feasible assignment can be found efficiently if one exists. (And otherwise, the algorithm can determine that no feasible assignment exists.)
This ceases to be true when both types of constraint are combined:

\begin{proposition} \label{prop:NP-hardness-topological}
  The following problem is NP-hard: given a set \CONSTRAINTS of constraints, each of which is of the form (1) $z_i$ must precede at least one element of $S_i$, or (2) $z_i$ must follow at least one element of $S_i$, decide whether there exists a linear order on the elements satisfying all the constraints.
  NP-hardness holds even when $\SetCard{S_i} \leq 3$ for all $i$.
\end{proposition}

\begin{proof}
  We give a straightforward reduction from 3SAT. Given an instance of $n$ variables $x_1, \ldots, x_n$ and $m$ clauses $C_j$, each of which is a disjunction of  three literals, the reduction constructs an instance with $2n+1$ elements and $n+m$ constraints.
  There are two elements $z_i, \bar{z}_i$ for each variable $x_i$, and one more element $\hat{y}$. For each variable $x_i$, there is a constraint that $\hat{y}$ must follow at least one of $z_i, \bar{z}_i$.
  For each clause $C_j = \ell_{j,1} \vee \ell_{j,2} \vee \ell_{j,3}$, there is a constraint that $\hat{y}$ must precede at least one of the elements corresponding to the $\ell_{j,i}$; for example, if $C_j = x_2 \vee \bar{x}_3 \vee x_5$, then the constraint says that $\hat{y}$ must precede at least one of the elements $\SET{z_2, \bar{z}_3, z_5}$.

  The correctness of the reduction can be easily observed by interpreting all the literals after $\hat{y}$ as true. The first set of constraints captures that no assignment may make both $x_i$ and $\bar{x}_i$ true, and the second set of constraints captures that each clause must have at least one true literal.
\end{proof}

We can now prove our main theorem for the simple algorithm:

\begin{theorem} \label{thm:iterations-simple-choice}
If the algorithm of \cref{prop:find-any-ranking} is used to compute the speculative preference order \PREF[t]{\AG} in \cref{choose-representative-many-to-many} of \cref{algo:learning-many-to-many-matching}, then \cref{algo:learning-many-to-many-matching} terminates after at most $O(n^{\MaxCapacity+2})$ iterations; furthermore, each iteration runs in time at most $O(n^{\MaxCapacity+3})$.
\end{theorem}

\begin{proof}
  The proof follows in the footsteps of that for Theorem~\ref{thm:learning-matching-time}. Assume that in iteration $t$, a blocking pair $\BLOCKING[t] = \SET{\WORKER,\FIRM}$ is revealed, implying that \PrefT{\WORKER}{\FIRM}{\FIRM[']} for some $\FIRM[']\in\Partners[t]{\WORKER}$ and \PrefT{\FIRM}{\WORKER}{\WORKER[']} for some $\WORKER[']\in\Partners[t]{\FIRM}$.
  Again, by the non-blocking property of a stable matching with respect to the \PREF[t]{\AG}, we can infer that at least one of \Pref[t]{\WORKER}{\Partners[t]{\WORKER}}{\FIRM} (i.e., \WORKER prefers each of \Partners[t]{\WORKER} to \FIRM) or \Pref[t]{\FIRM}{\Partners[t]{\FIRM}}{\WORKER} holds. Without loss of generality, \Pref[t]{\WORKER}{\Partners[t]{\WORKER}}{\FIRM} holds.
  In particular, because \PREF[t]{\WORKER} satisfies all constraints in \Constraints[t]{\WORKER}, we obtain that $\Constraint{\FIRM}{\Partners[t]{\WORKER}} \notin \Constraints[t]{\WORKER}$, implying that $\SetCard{\Constraints[t+1]{\WORKER}} = \SetCard{\Constraints[t]{\WORKER}} + 1$.

  Therefore, after each round, the quantity
  $\Phi_t = \sum_{\AG} \SetCard{\Constraints[t]{\AG}}$ increases by at least 1.
  The size of each constraint set \Constraints[t]{\AG} is upper-bounded by $\SetCard{\Constraints[t]{\AG}} \leq \binom{n}{\Capacity{\AG}} \cdot (n-\Capacity{\AG})=O(n^{q+1})$.
  This implies that $\Phi_t \leq \sum_{\AG} \binom{n}{\Capacity{\AG}} \cdot (n-\Capacity{\AG}) = O(n^{\MaxCapacity+2})$, so the process terminates after at most $O(n^{\MaxCapacity+2})$ iterations. 
   Similarly, because the algorithm of \cref{prop:find-any-ranking} runs in time $n\cdot\SetCard{\Constraints[t]{\AG}}=O(n^{q+2})$ for each agent (and since the running time of each iteration is dominated by this), each iteration runs in time at most $O(n^{q+3})$.
\end{proof}

For the one-to-one matching case, as discussed in the introduction, this proves a bound of $O(n^3)$ on the deterministic query complexity, which is obviously worse by a factor of $\nicefrac{n}{\log n}$ compared to the bound we obtained in Section~\ref{sec:one-to-one} with a more careful choice of speculative preference orders.

\subsection{General Quotas} \label{sec:general-capacities}
In this section, our goal is to improve the guarantee of Theorem~\ref{thm:iterations-simple-choice} to eliminate the (exponential) dependence on \MaxCapacity. As in Section~\ref{sec:one-to-one}, we do so by a more careful choice of the speculative preference orders \PREF[t]{\AG} used in Algorithm~\ref{algo:learning-many-to-many-matching}. Unlike in Section~\ref{sec:one-to-one}, our implementation will not be computationally efficient.

Our new sampling algorithm is a modification of \cref{algo:sample-preference}. It repeatedly estimates the fractions of valid preference orders (consistent with the given constraints) that have an element \AG in last position, among a set of remaining elements.
For any set \CONSTRAINTS of constraints, set \RemElements, and agent $\AG \in \RemElements$, let \LastFrac[\CONSTRAINTS]{\RemElements}{\AG} denote the fraction of preference orders consistent with \CONSTRAINTS that have \AG ranked last among the agents in \RemElements.
The algorithm, in each iteration $t$, will have a set \RemElements[t] of elements still in need of ranking, and chooses the last element based on the \LastFrac[\CONSTRAINTS]{\RemElements}{\AG} (or estimates thereof in Section~\ref{sec:many-to-many-sampling}).
It is detailed in Algorithm~\ref{algo:sample-ranking-many}.

\begin{algorithm}[htb]
\begin{algorithmic}[1]
  \STATE Let $\RemElements[t] = \SET{1, \ldots, n}$.
  \FOR{$t=n$ \textbf{downto} 1}
    \STATE Let $\LastElem[t] \in \argmax_{\AG \in \RemElements[t]}\LastFrac[\CONSTRAINTS]{\RemElements[t]}{\AG}$.
    \STATE Place \LastElem[t] in position $t$ of \PREFH.
    \STATE Let $\RemElements[t-1] = \RemElements[t] \setminus \SET{\LastElem[t]}$.
  \ENDFOR
  \STATE Return \PREFH.
\end{algorithmic}

\caption{The algorithm $\SampleRankingMany(\CONSTRAINTS)$. \label{algo:sample-ranking-many}}
\end{algorithm}

Akin to the analysis of representative preference orders in Section~\ref{sec:one-to-one}, we will show that the preference order \PREFH produced by \cref{algo:sample-ranking-many} guarantees a reasonable reduction in the number of remaining consistent preference orders.

\begin{lemma} \label{lem:many-to-many-progress}
  Let \PREFH be the preference order returned by \cref{algo:sample-ranking-many}. Let \Constraint{\AG[']}{S} be any constraint that holds in more than a $1\!-\!\frac{1}{n}$ fraction of preference orders consistent with \CONSTRAINTS.
  Then, \AG['] precedes at least one element of $S$ in \PREFH.
\end{lemma}

\begin{proof}
  Let \Constraint{\AG[']}{S} be any constraint that holds in more than a $1\!-\!\frac{1}{n}$ fraction of the preference orders consistent with \CONSTRAINTS.
  Let $t$ be the iteration of $\SampleRankingMany(\CONSTRAINTS)$ in which an element of $S \cup \SET{\AG[']}$ was placed for the first time, i.e., $t = \min \SetB{t'}{\LastElem[t'] \in S \cup \SET{\AG[']}}$.
  
  \begin{sloppypar}
  Because the algorithm chooses $\LastElem[t] \in \argmax_{\AG} \LastFrac[\CONSTRAINTS]{\RemElements[t]}{\AG}$, it ensured in particular that $\LastFrac[\CONSTRAINTS]{\RemElements[t]}{\LastElem[t]} \geq \frac{1}{n}$. In particular, because no element of $S$ had been previously placed, and since \Constraint{\AG[']}{S} was assumed to be a constraint holding in (strictly) more than a $1\!-\!\frac{1}{n}$ fraction of the preference orders consistent with~\CONSTRAINTS, we get that $\LastElem[t] \neq \AG[']$, so by definition of $t$, an element of $S$ must be placed in iteration $t$. This implies that \AG['] precedes at least one element of $S$, namely, \LastElem[t].
  \qedhere
  \end{sloppypar}
\end{proof}

\begin{theorem} \label{thm:learning-matching-many-iterations}
If $\SampleRankingMany(\Constraints[t]{\AG})$ is used to compute the speculative preference order \PREF[t]{\AG} in \cref{choose-representative-many-to-many} of \cref{algo:learning-many-to-many-matching}, then \cref{algo:learning-many-to-many-matching} terminates after at most $O(n^3 \log n)$ iterations.
\end{theorem}

\begin{proof}
As in the proof of Theorem~\ref{thm:learning-matching-iterations}, consider any iteration $t$ in which a blocking pair $\BLOCKING[t] = \SET{\WORKER,\FIRM}$ (with respect to the true preferences \PREFT{\AG}) was revealed.
Because the Gale-Shapley Algorithm finds a stable matching, the pair $\SET{\WORKER,\FIRM}$ was not blocking for the preferences \PREF[t]{\AG} used in running Gale-Shapley. In particular, this means that at least one of $\WORKER,\FIRM$ preferred all of their assigned matches over the other. Without loss of generality, assume that \Pref[t]{\WORKER}{\Partners[t]{\WORKER}}{\FIRM}, i.e., \Pref[t]{\WORKER}{\FIRM[']}{\FIRM} for all $\FIRM['] \in \Partners[t]{\WORKER}$).

By Lemma~\ref{lem:many-to-many-progress}, applied to \PREF[t]{\WORKER}, we infer that at least a $\frac{1}{n}$ fraction of the preference orders consistent with \Constraints[t]{\WORKER} have \Pref[t]{\WORKER}{\Partner[t]{\WORKER}}{\FIRM}. The revelation of the blocking pair $\SET{\WORKER,\FIRM}$ rules out all of the linear orders \PREFT{\WORKER} that have \PrefT{\WORKER}{\Partners[t]{\WORKER}}{\FIRM}. Therefore, $\NumRankings[t+1]{\WORKER} \leq (1-\frac{1}{n}) \cdot \NumRankings[t]{\WORKER}$. In turn, this implies that $\NUMRANKINGS[t+1] \leq (1-\frac{1}{n}) \cdot \NUMRANKINGS[t]$.

We obtain that Algorithm~\ref{algo:learning-many-to-many-matching} terminates after at most $T = \log_{\frac{n}{n-1}} (n!)^n = O(n^3 \log n)$ iterations.
\end{proof}

As in the case of Theorem~\ref{thm:learning-matching-iterations}, Theorem~\ref{thm:learning-matching-many-iterations} makes no statement about computational efficiency. Unlike in the case of one-to-one matchings, we are not aware of any efficient implementation. However, in the remainder of this section, we will show that as with one-to-one matchings, an ability to sample (nearly) uniformly preference orders from among all satisfying a given constraint set would be sufficient to achieve an efficient implementation. The obstacles to sampling are discussed in more depth in Section~\ref{sec:conclusions}.

\subsection{Estimations via Sampling} \label{sec:many-to-many-sampling}
Instead of the actual values \LastFrac[\CONSTRAINTS]{\RemElements[t]}{\AG}, Algorithm~\ref{algo:sample-ranking-many} can be run with estimates \LastFracE[\CONSTRAINTS]{\RemElements[t]}{\AG}. Specifically, in each iteration $t$, the algorithm can sample $K=6 n^2 \ln(n)$ linear orders \PREF[t]{i} independently and uniformly from the set of all linear orders satisfying all constraints \CONSTRAINTS. Then, \LastFracE[\CONSTRAINTS]{\RemElements[t]}{\AG} is the fraction of sampled orders that have \AG ranked last among \RemElements[t]. The algorithm uses the estimates \LastFracE[\CONSTRAINTS]{\RemElements[t]}{\AG} in place of the actual \LastFrac[\CONSTRAINTS]{\RemElements[t]}{\AG}. We will refer to this variant of \cref{algo:sample-ranking-many} as \SampleSampleRankingMany.

\begin{sloppypar}
We reiterate that we are not aware of any \emph{computationally efficient} algorithm to sample such~\PREF[t]{i}, even approximately uniformly. Nonetheless, we will analyze the impact of such sampling, and show that the remainder of the algorithm works as before.
\end{sloppypar}

We define (high-probability, as we will show) events \Event[t]{E}, which jointly will guarantee that the samples are sufficiently accurate. Specifically, \Event[t]{E} is defined as the event that all estimates \LastFracE{\RemElements[t]}{\AG} in round $t$ are accurate estimates to within at most $\pm \frac{1}{2n}$.
More formally, the event \Event[t]{E} is defined as $\AbsB{\LastFrac[\CONSTRAINTS]{\RemElements[t]}{\AG} - \LastFracE[\CONSTRAINTS]{\RemElements[t]}{\AG}} \leq \frac{1}{2n}$ for all $\AG \in \RemElements[t]$.
We begin by showing a suitable high-probability guarantee on these events.

\begin{lemma} \label{lem:last-count-accurate}
  For all $t$ and all sets \RemElements, we have that $\ProbC{\Event[t]{E}}{\RemElements[t] = \RemElements} \geq 1-\frac{2}{n^2}$.
\end{lemma}

\begin{proof}
  Consider any iteration $t$ and condition on $\RemElements[t] = \RemElements$. Let $\AG \in \RemElements$ be arbitrary.
  Because the \PREF{i} are sampled uniformly randomly, by definition, $\Expect{\LastFracE[\CONSTRAINTS]{\RemElements}{\AG}} = \LastFrac[\CONSTRAINTS]{\RemElements}{\AG}$.
  The Hoeffding Bound guarantees that
  $\ProbB{\Abs{\LastFracE[\CONSTRAINTS]{\RemElements}{\AG}-\LastFrac[\CONSTRAINTS]{\RemElements}{\AG}} \geq \frac{1}{2n}} = 2\exp(-2K/(2n)^2) \leq 2n^{-3}$.
  A union bound over all of the at most $n$ elements $\AG \in \RemElements$ now implies that all estimates are accurate to within an additive $\frac{1}{2n}$ with probability at least $1\!-\!\frac{2}{n^2}$.
\end{proof}

Define the event $\Event{E} = \bigwedge_{t=2}^n \Event[t]{E}$ to be the event that all relevant estimates in all iterations of $\SampleSampleRankingMany(\CONSTRAINTS)$ are accurate up to an additive $\frac{1}{2n}$. By the union bound, $\Prob{\Event{E}} \geq 1-\frac{2}{n}$.
The key implication of the event \Event{E} is captured by the following lemma, which is a slightly weaker version of Lemma~\ref{lem:many-to-many-progress}:

\begin{lemma} \label{lem:many-to-many-progress-sampling}
  Assume that the event \Event{E} happened, and let \PREFH be the preference order returned by $\SampleSampleRankingMany(\CONSTRAINTS)$. Let \Constraint{\AG[']}{S} be any constraint that holds in more than a $1\!-\!\frac{1}{2n}$ fraction of preference orders consistent with \CONSTRAINTS.
  Then, \AG['] precedes at least one element of $S$.
\end{lemma}

\begin{proof}
  The proof is nearly identical to that of \cref{lem:many-to-many-progress}. In the iteration when an element of $S \cup \SET{\AG[']}$ is placed for the first time, the algorithm uses the estimates \LastFracE{\RemElements[t]}{\AG} instead of \LastFrac{\RemElements[t]}{\AG}.
  Under the event \Event{E}, the fact that $\LastFrac{\RemElements[t]}{\AG[']} < \frac{1}{2n}$ by assumption and that the estimate is accurate to within an additive $\pm \frac{1}{2n}$ means that $\LastFracE{\RemElements[t]}{\AG[']} < \frac{1}{n}$, so again, \AG['] cannot be placed in this iteration.
\end{proof}  

We obtain the following theorem.

\begin{theorem} \label{thm:learning-many-to-many-matching-time}
If $\SampleSampleRankingMany(\Constraints[t]{\AG})$  (i.e., \cref{algo:sample-ranking-many} when run using random sampling-based estimates) is used to compute the preference order \PREF[t]{\AG} in \cref{choose-representative-many-to-many} of \cref{algo:learning-many-to-many-matching}, then \cref{algo:learning-many-to-many-matching} terminates after at most $O(n^3 \log n)$ iterations with high probability; furthermore, each iteration requires drawing $O(n^2 \log n)$ samples for $n$ iterations of the \textbf{for} loop of \cref{algo:sample-ranking-many}, for each of the $n$ agents. Thus, each iteration requires drawing $O(n^4 \log n)$ samples.
\end{theorem}

\begin{proof}
  The proof is practically identical to that of Theorem~\ref{thm:learning-matching-time}. Because \Event{E} has high probability, if the algorithm is run for $\Theta(n^3 \log n)$ iterations, the number of iterations in which its speculative preference orders are such that they sufficiently reduce the number of candidate scenarios is also $\Theta(n^3 \log n)$. The other iterations do not guarantee progress, but cannot increase the number of remaining scenarios. Thus, up to constant, we achieve the same guarantees as the deterministic Theorem~\ref{thm:learning-matching-many-iterations}.
\end{proof}

\section{Conclusions} \label{sec:conclusions}
We presented interactive algorithms for learning stable matchings, both in the case of one-to-one matchings and in the case of many-to-many matchings. The algorithm repeatedly proposes a matching; if the matching is not stable, the algorithm is informed of one blocking pair. We showed that in the case of one-to-one matchings, there is a computationally efficient interactive algorithm that finds a stable matching after at most $O(n^2 \log n)$ rounds of interaction with high probability.
We also showed that any algorithm for finding a stable matching in this model requires $\Omega(n^2)$ rounds of interaction, even in expectation.

Our algorithm extends to the case of many-to-many matchings, where the guarantee on the number of rounds deteriorates to $O(n^3 \log n)$. However, this algorithm is not computationally efficient; an ``efficient'' version can be constructed which requires $O(n^{\MaxCapacity+2})$ iterations when each agent can be matched with at most \MaxCapacity other agents.

The most immediate open question is whether the algorithm with $O(n^3 \log n)$ rounds of interactions can be implemented in polynomial time. A sufficient condition for doing so would be to be able to sample (approximately) uniformly from the set of all linear orders consistent with a set of constraints of the form ``element $z$ must precede at least one element from the set $S$.''
This is an interesting sampling question in its own right, and it also has applications in other scenarios in which a principal is trying to learn preference rankings of agents from observed behavior, e.g., in the case of learning valuations of single-minded buyers from their product choices.

This sampling problem appears non-trivial. As we discussed (and heavily used), the special case when all sets $S$ have size $1$ is exactly the problem of sampling a linear extension of a given partial order. This is accomplished by the well-known Karzanov-Khachiyan chain \citep{karzanov:khachiyan:conductance}, which starts from any feasible linear extension (e.g., computed by Topological Sort), and repeatedly picks a uniformly random index $i$, swapping elements $i$ and $i+1$ with probability~\nicefrac{1}{2} if doing so is consistent with the partial order constraints, and doing nothing otherwise. \citeauthor{karzanov:khachiyan:conductance} show rapid mixing of this chain by studying conductance properties of the order polytope associated with the partial order \citep{stanley:poset-polytopes}. \citet{bubley:dyer:linear-extensions} show that giving a slight bias towards sampling indices $i$ towards the middle of the ranking leads to a much simpler analysis based on a beautiful and elementary path coupling.
One could still aim to prove rapid mixing for the Karzanov-Khachiyan chain, applied to constraints in which the sets $S$ are larger than singletons.
Unfortunately, the natural representation of constraints in $n$-dimensional space ceases to be convex, so it is not clear whether the conductance arguments can be applied. Similarly, the specific path coupling of \citeauthor{bubley:dyer:linear-extensions} also seems to break down for more complex constraints.

An open question with fewer implications, but possible technical interest, is whether the number of iterations can be reduced from $O(n^3 \log n)$ to $O(n^2 \MaxCapacity \log n)$ when all agents' capacities are bounded by \MaxCapacity. One would not expect a better bound, because even when the rankings are not constrained at all (e.g., in the first iteration), a constraint of the form \Constraint{\AG}{S} only rules out a $\frac{1}{\SetCard{S}+1}$ fraction of rankings.
A sufficient result would be a generalization of Yu's Proportional Transitivity Theorem \citep{yu:proportional-transitivity}. A possible formulation of such a result would be the following.
For a given set \CONSTRAINTS of constraints of the form \Constraint{\AG[i]}{S_i}, where each $\SetCard{S_i} \leq \MaxCapacity$, let \Closure{\phi}{\CONSTRAINTS} be the set of all constraints \Constraint{\AG}{S} with $\SetCard{S} \leq \MaxCapacity$ such that at least a $\phi$ fraction of all rankings \PREFE consistent with \CONSTRAINTS satisfies \Constraint{\AG}{S}.
Theorem~\ref{thm:proportional-transitivity} can then be restated as saying that when $\MaxCapacity=1$, for every $\phi\geq \phi_{\min}$, we have that $\Closure{1}{\Closure{\phi}{\CONSTRAINTS}} \subseteq \Closure{\phi}{\CONSTRAINTS}$.
A possible extension of Yu's Theorem would then state that for every \MaxCapacity, there is a $\phi_{\min}^{\MaxCapacity} = 1-\Omega(\nicefrac{1}{\MaxCapacity})$ such that when each constraint $\Constraint{\AG}{S} \in \CONSTRAINTS$ has $\SetCard{S} \leq \MaxCapacity$, for every $\phi\ge \phi_{\min}^{\MaxCapacity}$ the constraint set \CONSTRAINTS satisfies
$\Closure{1}{\Closure{\phi}{\CONSTRAINTS}} \subseteq \Closure{\phi}{\CONSTRAINTS}$.

Proving such an extension --- if it were true --- would likely require a very different approach from that of \citet{yu:proportional-transitivity}. \citeauthor{yu:proportional-transitivity}'s proof is based on an application of \citeauthor{shepp:xyz-inequality}'s XYZ inequality \citep{shepp:xyz-inequality}, combined with a dimension reduction technique of \citet{ball:logarithmically-concave} (see also a similar proof in \citet{kahn:yu:log-concave}). \citeauthor{ball:logarithmically-concave}'s result very heavily relies on the set under consideration (here, the order polytope) being convex, since it relies on the Brunn-Minkowski Inequality. As discussed in the preceding paragraphs, the natural $n$-dimensional representation of larger constraints ceases to be convex, so one of the key techniques will fail to be available.

\bibliographystyle{ACM-Reference-Format}
\bibliography{davids-bibliography/names,davids-bibliography/conferences,davids-bibliography/publications,davids-bibliography/bibliography,local-bibliography}

\appendix

\section{Inapplicability of the Binary Search in Graphs Framework}
\label{sec:binary-search}
In \citet{OnlineLearning}, a framework was proposed for analyzing the number of rounds necessary to learn combinatorial structures in interactive learning settings such as ours. In the framework of \citet{OnlineLearning}, which generalizes \citeauthor{angluin:queries-concept}'s \cite{angluin:queries-concept} Equivalence Query Model, an algorithm repeatedly gets to propose a ``structure,'' and either learns that it has proposed the correct structure, or is given a small ``local'' correction to the structure. The key notion in \citet{OnlineLearning} is their Definition~2.1, restated here as follows:

\begin{definition} \label{def:graph-feedback-model}
Let $G$ be a (directed or undirected) graph whose nodes are the $N$ candidate structures and whose edges have positive weights. For each node/structure $v$, let $R_v$ be the set of possible responses that the learner can receive, and let $\Gamma_v$ denote the neighborhood of $v$ in $G$.
It is required that for each $v$, there be a mapping $\phi: R_v \to \Gamma_v$ with the following property: if $t$ is the correct structure (the ``target'') and $r$ is a correct response to the structure $v$ when $t$ is the target, then $\phi(r)$ must lie on a shortest path from $v$ to $t$ in $G$.
\end{definition}

Using machinery from \citet{BinarySearch}, it is shown in \citet{OnlineLearning} that under the feedback model of Definition~\ref{def:graph-feedback-model}, the correct structure can be learned in few rounds of interactions. Specifically, assume that the structure is known to lie in a set $S$ of $N_0 \leq N$ of the candidate structures.
If $G$ is undirected, then the correct structure can always be learned in at most $\log_2 N_0$ rounds of interaction.
This result furthermore extends to some directed graphs: if each edge $e$ (of weight $w_e$) is part of directed cycle of total weight at most $c w_e$, then the correct structure can be learned in at most $\log_{\frac{c}{c-1}} N_0$ rounds of interaction. (Notice that undirected graphs are subsumed by the special case $c=2$.)

The proof of this result relies heavily on Proposition~2.1 of \citet{OnlineLearning}, the relevant special case of which we restate here as follows:
\begin{proposition}[Proposition~2.1 of \citealp{OnlineLearning}]
  \label{prop:approximate-median}
  Let $G$ be a (weighted, directed) graph satisfying Definition~\ref{def:graph-feedback-model}, and assume that each edge $e$ of weight $w_e$ is part of a cycle of total weight at most $c w_e$. Then, for every set $S \subseteq V(G)$ of nodes, there exists a node $v=v(S)$ such that for every valid response $r \in R_v$, at most a $\frac{c-1}{c}$ fraction of the nodes in $S$ have a shortest path from $v$ to them starting with the edge $\phi(r)$.
\end{proposition}

While our interaction model at a high level is exactly that of \citet{OnlineLearning}, we show that for finding stable matchings, there is \emph{no} useful graph $G$ (undirected or directed) satisfying Definition~\ref{def:graph-feedback-model}. Thus, the techniques of \citet{OnlineLearning} cannot be leveraged to obtain our $O(n^2 \log n)$ result from Theorem~\ref{thm:learning-matching-time}. We will show this result even for the simplest case of one-to-one matchings.

First, in applying the framework of \citet{OnlineLearning}, it is important to clarify which structures would comprise the nodes of $G$. While the algorithm proposes a matching in each round, the revelation of a blocking pair does not constitute a local correction or improvement of the matching itself, but provides information about the underlying \emph{preferences}. Therefore, the set of all candidate structures is the set of all $2n$-tuples of $n$-element preference orders, i.e., there are $(n!)^{2n}$ structures. Because there are only $n!$ possible matchings, by symmetry, each matching is stable for $(n!)^{2n-1}$ structures, and the algorithm succeeds when it guesses any one of them. The feedback is always consistent with one particular (adversarially chosen) structure.

In each round, a blocking pair reveals to the algorithm exactly one (not necessarily adjacent) transposition in one or two of the preference orders. We show that even the following, apparently simpler, interactive search problem cannot be encoded in a graph satisfying Definition~\ref{def:graph-feedback-model}.

\begin{proposition} \label{prop:transposition-impossibility}
  Consider the following interactive learning problem: 
  The set of structures is the set of all preference orders over $n$ elements. There is an (unknown to the algorithm) correct preference order $\succ^*$. In each round $t$, the algorithm proposes a preference order $\succ^{(t)}$. Unless $\succ^{(t)}$ equals $\succ^*$, the algorithm is shown one (not necessarily adjacent) pair of elements that $\succ^{(t)}$ and $\succ^*$ have in opposite orders.

  There is no (weighted) undirected graph $G$ for this feedback model satisfying Definition~\ref{def:graph-feedback-model}. 
  There is also no (weighted) directed graph $G$ for this feedback model in which each edge of weight $w_e$ is part of a directed cycle of total weight at most $c w_e$, for any $c < n$.
\end{proposition}

Notice the contrast between Proposition~\ref{prop:transposition-impossibility} and Lemma~3.1 of \citet{OnlineLearning}, which states that when the pair of elements that is revealed is always \emph{adjacent}, there exists an undirected (and unweighted) graph satisfying Definition~\ref{def:graph-feedback-model}, leading to an efficient learning algorithm.

\begin{proof}
  We prove the proposition by showing that it violates the conclusion of Proposition~\ref{prop:approximate-median} for any $c < n$.
  Specifically, let $S$ be the set of all $n$ cyclical shifts of $(1, 2, \ldots, n)$, i.e., the orders $\PREF{1} = [1 \succ 2 \succ 3 \succ \cdots \succ n],
  \PREF{2} = [2 \succ 3 \succ \cdots \succ n \succ 1],
  \PREF{3} = [3 \succ 4 \succ \cdots \succ n \succ 1 \succ 2], \ldots,
  \PREF{n} = [n \succ 1 \succ 2 \succ \cdots \succ n-1]$.
  Let $v$ denote the node corresponding to \emph{any} linear order \PREFH of $\SET{1, \ldots, n}$ (not necessarily in $S$). We will focus on the response $r\in R_v$ defined as follows. If there exists any index $i$ such that \PrefH{i\!+\!1}{i}, then $r$ reveals any such pair $(i\!+\!1,i)$, informing the algorithm that $i$ should precede $i\!+\!1$. Notice that the only order in $S$ not consistent with this information is~\PREF{i}. If no such $i$ exists, then \PREFH must equal \PREF{1}. In this case, $r$ reveals that $n$ should precede $1$; this is consistent with all of $S$ except \PREF{1}.

  We have shown that there is a set $S$ of size $n$ such that for each proposed node $v$ corresponding to a linear order \PREFH, there exists a valid response $r$ such that an $\frac{n-1}{n}$ fraction of $S$ is consistent with $r$. Hence, any (weighted, directed) graph $G$ satisfying Definition~\ref{def:graph-feedback-model} must have some edge $e$ of weight $w_e$ that is not part of any cycle of total weight less than $n w_e$.
\end{proof}

In particular, applying the techniques of \citet{OnlineLearning} cannot guarantee a better bound than $\log_{\frac{n}{n-1}} (n!)^{2n} = \Omega(n^3 \log n)$ for the number of rounds. Notice that a bound of $O(n^3 \log n)$ is in fact worse than even the bound of $O(n^3)$ obtained by the simple $O(n^3)$ algorithm for one-to-one matchings discussed in \cref{sec:introduction}. This confirms that the techniques of \citet{OnlineLearning} cannot be applied to yield a more direct proof of our result, and the machinery of proportional transitivity appears necessary to obtain our results.

\section{From Partial Preference Lists to Full Preference Lists}
\label{sec:partial-to-full}
In this \lcnamecref{sec:partial-to-full}, we explain why by a well-known reduction in stable matching theory, restricting attention to markets with full preference lists is without loss of generality. Given any market with partial preference lists, for each agent \AG with quota $\Capacity{\AG}$ we will add to the other side of the market \Capacity{\AG} new ``phantom'' agents, each with a quota of $1$. These agents will rank \AG first, and then rank all other agents on \AG's side of the market (real and phantom) in arbitrary order. We complete the preference order of \AG by appending to it, after all agents acceptable to \AG, the \Capacity{\AG} phantom agents of \AG, followed by all other agents on the other side of the market in arbitrary order; these are the other phantom agents as well as real agents that are unacceptable to \AG.

Under this construction, it is well known that there is a bijection between the set of matchings \GMATCH in the original (partial preference lists) market and the set of perfect matchings $\GMATCH'$ in the completed (full preference lists) market in which no pair of phantom agents block. This bijection is as follows: $\GMATCH'\supseteq\GMATCH$, and in addition, in $\GMATCH'$ each agent \AG is matched with the $\Capacity{\AG}-\SetCardB{\Partners{\AG}}$ of its phantom agents that it prefers the most. (The remaining, phantom, agents are matched among themselves to create no blocking pairs between them.) Under this embedding, \GMATCH is stable if and only if $\GMATCH'$ is stable, a pair $(\WORKER,\FIRM)\in\WORKERS\times\FIRMS$ is blocking in \GMATCH if and only if it is blocking in $\GMATCH'$, and an agent \AG individually blocks~\GMATCH if and only if it blocks $\GMATCH'$ with one of its phantom agents (as a blocking pair).

Given the above, it is now clear why an algorithm designed to work only with markets with full preference lists can be applied to interactively learn a matching also in markets with (possibly) partial preference lists. Consider the following intermediary protocol between this algorithm $\mathcal{A}$ and the environment with partial preference lists: the intermediary will execute $\mathcal{A}$, aiming to learn a stable matching in the ``completion'' of the market of the environment to full preference lists. When $\mathcal{A}$ outputs a perfect matching $\GMATCH'$ in the completed market, the intermediary will present to the environment the corresponding matching \GMATCH in the original market.\footnote{We can ensure that $\mathcal{A}$ outputs a matching $\GMATCH'$ corresponding to some matching \GMATCH in the original market, by revealing to $\mathcal{A}$ all of the preferences of all phantom agents in advance. For our algorithms, this guarantees that no pairs of phantom agents will block $\GMATCH'$, as required.} Any blocking pair or individually blocking agent revealed by the environment will then be translated by the intermediary to a blocking pair with respect to \GMATCH, and revealed to $\mathcal{A}$. Once $\mathcal{A}$ learns a stable matching $\GMATCH'$ in the completed market, the intermediary will have learned a stable matching in the original market, as desired.\footnote{While we may have added many phantom agents, since their preferences are revealed to $\mathcal{A}$ in advance, the query complexity will not deteriorate because of the added agents.}

For the above reason, in our analysis, we assume without loss of generality that all preference lists are full, and we construct learning algorithms that only output perfect matchings. Therefore, if the proposed matching is unstable, then a blocking pair is revealed. We conclude this \lcnamecref{sec:partial-to-full} with an observation that shows that one should be careful with applying even well-known reductions within the context of our coarse query model.

We have shown that any algorithm for full preference lists can be converted into an algorithm for partial preference lists. While any algorithm for partial preference lists can also be converted to work with full preference lists in the adversarial model, we observe that this is in fact not true in other models, due to the requirement to always propose perfect matchings. To see this, consider the result of \citet{roth:vandevate:paths} which we mentioned in \cref{sec:introduction}, which states that in a model where the response blocking pair is chosen randomly among all blocking pairs based on a myopic (i.e., which depends only on the proposed matching) distribution that gives positive probability for any of the blocking pairs, with probability one, the myopic algorithm converges. As already mentioned in passing, this result crucially depends, even for markets with full preference lists, on the ability to resolve the blocking pair by matching its agents and leaving their previous partners unmatched. Indeed, \citet{tamura:stable} showed that if their partners were to be rematched with each other immediately (as in \citeauthor{knuth:stable}'s original version of the myopic algorithm), then the algorithm would never converge from some initial (perfect) matchings --- the existence of a ``path'' from these matchings to a stable matching depends on the ability to leave the former partners of a blocking pair unmatched. \citet{tamura:stable} also shows that to fix this problem, an additional type of query ---  identifying certain cycles --- is required.

In our adversarial setting, though, this issue could not have arisen: since in a matching with full preference lists, the two former partners of a resolved blocking pair together form a new blocking pair, the environment can simply report this pair (and force it to match if the algorithm is myopic) before continuing. More generally, in an adversarial setting, the requirement to always propose perfect matchings does not restrict the power of any learning algorithm since even without this requirement, the adversary can always reveal individually blocking agents if they exist, and thus never reveal any new information for matchings that are not perfect. (On the other hand, a randomly, rather than adversarially, chosen blocking pair may indeed contain new information, which is what saves the myopic algorithm of \citet{roth:vandevate:paths} from failing.)

\end{document}